\documentclass[a4paperb,11pt]{article}
\usepackage[utf8]{inputenc}
\usepackage[british]{babel}
\usepackage[mono=false]{libertine}
\usepackage[T1]{fontenc}
\usepackage{amsthm}
\usepackage[top=2.5cm, bottom=2.5cm, left=2cm, right=2cm]{geometry}
\usepackage{xcolor}
\usepackage[font=small]{caption}
\usepackage{enumitem}
\usepackage{cite}
\usepackage[pdftex]{graphicx}
\DeclareGraphicsExtensions{.pdfsx,.jpeg,.png}
\usepackage[cmex10]{amsmath}
\usepackage{array}
\usepackage{url} 
\usepackage{mathtools}
\usepackage{amssymb,amsfonts}
\usepackage{dsfont}
\usepackage{stmaryrd}
\usepackage{bm}
\usepackage{citesort}
\usepackage{color}
\usepackage{algorithm}
\usepackage[noend]{algpseudocode}
\usepackage{titlesec}
\newtheorem{example}{Example}
\newtheorem{condition}{Condition}
\makeatletter
\newcommand{\setappendix}{Appendix~\thesection:~~}
\newcommand{\setsection}{\thesection~~}
\titleformat{\section}{\bfseries\LARGE}{%
	\ifnum\pdfstrcmp{\@currenvir}{appendices}=0
	\setappendix
	\else
	\setsection
\fi}{0em}{}
\makeatother
\usepackage[titletoc]{appendix}
\usepackage[pdfpagemode=UseNone,bookmarksopen=false,colorlinks=true,urlcolor=blue,citecolor=magenta,citebordercolor=blue,linkcolor=blue]{hyperref}
\usepackage{setspace}
\usepackage{notations}

\DeclareUnicodeCharacter{00A0}{~}

\def \({\left(}
\def \){\right)}
\def \[{\left[}
\def \]{\right]}
\newcommand{\bsig}{{\boldsymbol{\sigma}}}

\newcommand{\bu}{{\textbf {u}}}

\newcommand{\bx}{{\textbf {x}}}

\newcommand{\btau}{{\boldsymbol{\tau}}}

\newcommand{\<}{\langle}
\renewcommand{\>}{\rangle}

\newcommand{\be}{\begin{equation}}
\newcommand{\ee}{\end{equation}}

\newcommand{\bea}{\begin{align}}
\newcommand{\eea}{\end{align}}

\newcommand{\tJ}{J}

\newcommand{\sH}{\mathcal{H}}

\newtheorem{theorem}{Theorem}
\newtheorem{lemma}[theorem]{\textbf{Lemma}}
\newtheorem{thm}[theorem]{\textbf{Theorem}}

\newtheorem{proposition}[theorem]{\textbf{Proposition}}

\DeclareMathAlphabet{\varmathbb}{U}{bbold}{m}{n}

\newcommand*{\QED }{\hfill\ensuremath{\square}}
\begin{document}
\title{Concentration of multi-overlaps for random ferromagnetic spin models}
\author{Jean Barbier$^{\star*}$, Chun Lam Chan$^{\dagger}$ and Nicolas Macris$^{\dagger}$}
\date{}
\maketitle
{\let\thefootnote\relax\footnote{
\hspace{-18.5pt}
$\star$ Quantitative Life Sciences, The Abdus Salam International Center for Theoretical Physics, Trieste, Italy.\\
$*$ Statistical Physics Laboratory, \'Ecole Normale Sup\'erieure, Paris, France.\\
$\dagger$ Communication Theory Laboratory, \'Ecole Polytechnique F\'ed\'erale de Lausanne, Switzerland.
}}
\begin{abstract}
We consider ferromagnetic spin models on dilute random graphs and prove that, with suitable one-body infinitesimal perturbations added to the Hamiltonian, the multi-overlaps 
concentrate for all temperatures, both with respect to the thermal Gibbs average and the quenched randomness. Results of this nature have been known only for the lowest order overlaps, 
at high temperature or on the Nishimori line.
Here we treat all multi-overlaps by a non-trivial application of Griffiths-Kelly-Sherman correlation inequalities. Our results apply in particular to 
the pure and mixed $p$-spin ferromagnets on random dilute Erdoes-R\'enyi hypergraphs. On physical grounds one expects that multi-overlap concentration directly implies the correctness of the cavity 
(or replica symmetric) formula for the pressure. The proof of this formula for the general $p$-spin ferromagnet on a random dilute hypergraph remains an open problem. 
\end{abstract}
{
	\singlespacing
	\hypersetup{linkcolor=black}
	\tableofcontents
}
\section{Introduction}

Disordered spin models have a long history and despite much progress their statistical mechanics is far from understood. Among them, mean-field models have a 
special status because they can often be solved exactly by means of replica or cavity methods \cite{mezard1990spin}. These solutions are generally 
controlled by a set of ``multi-overlap parameters'' (simply called ``overlaps'' in the following). If $\sigma_i=\pm 1$ for $i=1, \dots, n$ denote a set of $n$ binary spins of the model, the overlap parameters are generally defined as 
$Q_k \equiv \frac{1}{n}\sum_{i=1}^n \sigma_i^{(1)}\sigma_i^{(2)}\dots \sigma_i^{(k)}$
where $k\geq 1$ is an integer and $(\sigma_i^{(\alpha)})_{i=1,\dots, n}^{\alpha=1,\dots, k}$ are distributed according to the replicated Gibbs distribution, 
in other words the product of $k$ copies of the Gibbs distribution. It is believed that the distribution of the overlaps controls 
the nature of the solution and in particular whether it displays a
``replica symmetric'' behaviour or a more complicated sequence of ``replica symmetry broken'' phases. In the replica 
symmetric phase the overlaps are concentrated on single numerical values, while in the replica symmetric broken phase 
they do not self-average and their distribution displays non-trivial ultrametric properties \cite{mezard1990spin}. 

Much progress on this picture has been done in the relatively recent literature,
starting from pioneering results of Pastur and Scherbina \cite{PasturScherbina:1991,Shcherbina:1991}, 
Ghirlanda and Guerra \cite{ghirlanda1998general,Guerra:2006:intro,Guerra:2006}, and Aizenman and Contucci \cite{AizenmanContucci:1998}.
We refer to \cite{Chatterjee:2014:superconcentration,talagrand2010meanfield1,talagrand2010meanfield2,Panchenko2013} and references therein for more recent progress as well as comprehensive reviews. 

Most results in the literature concern models defined on the complete graph with quenched random coupling constants, 
e.g. the Sherrington-Kirkpatrick model and generalizations to $p$-spin interactions with $p\geq 3$, as well as their spherical incarnations. 
For such models, on the {\it complete} graph, the whole statistical mechanical solution is generally controlled by the lowest order 
overlaps, namely $Q_1$ and $Q_2$, and known results mainly concern $Q_1$, $Q_2$ (note that $Q_1$ is related to the magnetization and $Q_2$ to the Edwards-Anderson order parameter). 

In contrast, on {\it dilute} graphs (typically locally tree-like sparse Erdoes-R\'enyi random hpergraphs)
the {\it whole} sequence of overlaps $\{Q_k, k\geq 1\}$ controls the replica solution, and the corresponding study 
of their distribution for all $k\geq 1$ is largely open. 
Note that \cite{talagrand2010meanfield1,talagrand2010meanfield2} discusses the high temperature case for the lowest order overlaps of some dilute graph models. Another interesting 
recent piece of work is the one by Chatterjee on the random field Ising model over generic graphs, and in any dimensions \cite{Chatterjee2015}. 
There, again, only the lowest order overlap is controlled but for (almost) all temperatures and external field strenghts. 
The author claims that his result implies the replica symmetry of the model. We believe that in order for 
such a claim to hold in the most generic sense one should in addition show that such concentration results imply the validity of the 
replica symmetric fomula for the free energy. This is a non-trivial task even for models on dense graphs (see the discussion at the end of the introduction) 
and for dilute models it is probably necessary to control overlaps of all orders.

In \cite{barra2006overlap,barra2007stability,FranzLeoneToninelli:2003,LucaSilvio:2009}
non-trivial constraints analogous to the Aizenman-Contucci and Ghirlanda-Guerra ones are derived for all overlaps of dilute models but one cannot deduce the concentration properties from 
these (nevertheless some of the techniques used in the present contribution are inspired from these works).

In this contribution our main interest is the study of fluctuations and self-averaging properties of {\it all} overlaps $\{Q_k, k\geq 1\}$ for the Ising multi-body 
ferromagnets on sparse random hypergraphs (typically of Erdoes-R\'enyi type).
We distinguish two types of 
fluctuations, namely the thermal ones and those with respect to the disorder, informally measured by the two quantities
$\mathbb{E} \langle (Q_k - \langle Q_k \rangle)^2\rangle $ and 
$\mathbb{E}[ (\langle Q_k\rangle  - \mathbb{E}\langle Q_k \rangle)^2 ]$.
Adding these two fluctuations one finds the total fluctuations $\mathbb{E} \langle (Q_k - \mathbb{E}\langle Q_k \rangle)^2\rangle $. 
Here 
$\langle -\rangle$ denotes the Gibbs expectation with respect to the ``replicated'' Gibbs measure and $\mathbb{E}$ the 
expectation with respect to the quenched disorder (see Section \ref{ferromodelmulti-overlapconcentration} for the precise definitions). 
Our main result states that both types of fluctuations vanish in the thermodynamic limit for 
all temperatures. For this result to hold at all temperatures, we must add suitable ``infinitesimal'' one-body perturbations to the Hamiltonian\footnote[1]{By infinitesimal perturbations we mean perturbations that 
do not change the thermodynamic limit of the pressure when we take the limit of zero perturbation after the thermodynamic limit.}. 
Indeed, concentration may hold only within a ``pure state'' and it is well known that one must add suitable perturbations in order to 
select pure states (that may coexist at low temperatures). 

We would like to stress that, for disordered systems, the nature of the perturbation that one should add is not always clear. 
For example all $p$-spin interactions are sometimes added to the two-body Sherrington-Kirkpatrick Hamiltonian, or even in 
the finite-dimensional short-range Edward-Anderson model, and it is perhaps not so clear what the physical interpretation of such perturbations is \cite{Panchenko2013}. 
Here we limit ourselves to the simple one-body perturbations that can physically be interpreted as infinitesimal external magnetic fields.

To the best of our knowledge this is the first time a concentration result is established for {\it all} overlaps $\{Q_k, k\geq 1\}$ in a dilute disordered spin model 
for all temperatures. Examples of models that are covered by our results are the pure and mixed $p$-spin dilute ferromagnets on 
random Erdoes-R\'enyi hypergraphs\footnote{With minor adjustments in the formulation of the models we can also cover ferromagnets 
on dense graphs.}.
Here the coupling constants are ferromagnetic and this allows the use of the Griffith-Kelly-Sherman (GKS) inequality which plays an important role in our analysis. 
There are at least two interesting directions where one might hope to extend this type of result. First, the $p$-spin dilute ferromagnet in a random external magnetic field 
(taking both negative and positive signs), as considered in \cite{Chatterjee2015} for $p=2$. In fact many of our intermediate results still hold with almost identical proofs 
by using the Fortuin-Kasteleyn-Ginibre (FKG) inequality instead of GKS. A combination of FKG and the Ghirlanda-Guerra identities actually forms the basis of the analysis in \cite{Chatterjee2015} for the lowest order overlap. Second, dilute $p$-spin models with random coupling constants (taking both signs) on their Nishimori line. These models have important applications, for example in communications \cite{tanaka2002statistical,korada2010,Montanari2005,MacrisKudekar2009} or in community detection (see \cite{CIT-067} and references therein).

Let us briefly discuss possible applications of our analysis, to which we hope to come back in the future. It is known folklore that the self-averaging of the overlaps should 
imply that the solution of the model is replica symmetric, and in particular that the replica symmetric expression of the thermodynamic pressure is valid. However, to our knowledge, 
this logical implication has never been mathematically settled in a clear way. In \cite{PasturScherbina:1991} the authors  
show (for models on complete graphs) that if the pressure is not given by a replica symmetric expression then the overlaps cannot concentrate. Also, the Guerra-Toninelli interpolation method \cite{guerra2005introduction}
makes it clear that if the pressure is not replica symmetric then the overlaps of an ``interpolated model'' cannot concentrate. Recently two of us  have developed 
an ``adaptive interpolation method'' \cite{barbier_stoInt} (a powerful extension of the Guerra-Toninelli method)
which deduces the replica symmetric expression of the pressure from the total concentration of the overlaps 
of an interpolated model with suitable perturbation added. Here it is important to stress that (currently at least) we need to control both type of fluctuation introduced above. This adaptive interpolation method has been applied with success to 
many models on complete graphs \cite{2017arXiv170910368B,Barbier:PMLR2018,toappear,Aubin:NIPS2018,Bengio:NIPS2018} where concentration of the lowest order overlap suffices (and the solution is replica symmetric for all temperatures). It is still largely an open problem to extend the method to situations on dilute 
graphs where all overlaps must be analyzed. The present paper is a first step towards developing this program for the $p$-spin Ising ferromagnet on dilute hypergraphs. 

We point out that the replica symmetric (or cavity) formula has been proved for the case $p=2$ by different methods \cite{Dembo-Montanari:2010} using the GKS and Griffiths-Hurst-Sherman (GHS) inequalities as well as algorithmic message passing ideas. Other related models where the replica symmetric formula for the pressure is established are $p$-spin models on dilute graphs with random couplings on their Nishimori line \cite{CojaOghlanKPZ:2018}. This work uses the Aizenman-Sims-Starr approach to the cavity method \cite{aizenman2003extended}, in conjunction with the 
Guerra-Toninelli interpolation method first developed for dilute spin models in \cite{FranzLeone} and applied to communication models in \cite{Montanari2005,korada2010,MacrisKudekar2009,PanchenkoTalagrand2004} or used for systematically proving the existence of the thermodynamic limit of various quantities in spin models and constraint satisfaction problems in \cite{bayati2010combinatorial,salez2016interpolation}.
When the Nishimori symmetry is present, concentration for the thermal fluctuations has already been established (and used in these works) in various guises 
\cite{MacrisKudekar2009,Montanari:ETT2008,CojaOghlanKPZ:2018} by using special identities implied by the Nishimori symmetry. However a proof of concentration for fluctuations with respect to the quenched disorder and full concentration of all overlaps is still elusive. 

In Section \ref{ferromodelmulti-overlapconcentration} we formulate the models and state our main theorems. The proofs are found in Section \ref{sec:conc}. The appendices contain technical intermediate results.

\section{Ferromagnetic spin models and overlap concentration}\label{ferromodelmulti-overlapconcentration}
Consider a collection of $n$ binary spins $\sigma_i\in \{-1, 1\}$, $i=1,\dots, n$. For any subset $X\subset\{1,\dots,n\}$  we denote
$\sigma_X = \prod_{i\in X} \sigma_i$.
A generic ferromagnetic spin system has Hamiltonian 
\begin{align}\label{mostgeneralhamiltonian}
 \sH_{0}(\bsig) \equiv - \sum_{X\subset \{1,\dots,n\}} \tJ_X\sigma_X
\end{align}
where $J_X\geq 0$ and the sum runs over all possible $2^n$ subsets of $\{1,\dots,n\}$. The only subsets of spins that truly participate in the interactions are of course those for which $J_X > 0$. 
The random models that we consider have independently distributed coupling constants $J_X$, $X\subset \{1,\dots,n\}$, with distribution supported on $\mathbb{R}_{\ge 0}$. As already said in the introduction our main interest is in dilute systems, a typical example of which is given below.

The thermodynamic potential of interest is the pressure
\begin{align*}
P_n \equiv \frac{1}{n}\ln {\cal Z} = \frac{1}{n}\ln \sum_{\bsig\in\{\pm1\}^n}\exp(-\sH_{0}(\bsig))
\end{align*} 
where ${\cal Z}$ is the partition function of the model. Without loss of generality we consider the inverse temperature $\beta =1$ as this amounts to a simple global rescaling of the interactions.
The average pressure is defined as $p_n \equiv \mathbb{E}\, P_n$ where $\mathbb{E}$ is the expectation over all the coupling constants.
For models of physical interest one expects that $P_n$ concentrates over $p_n$. Our theorems on overlap concentration stated below are 
formulated in a generic setting and hold as long as the concentration of the pressure holds:
\begin{align}\label{concpressuregeneral}
 \mathbb{E}[( P_n - p_n)^2] \leq \frac{C_P}{n}
\end{align}
for $C_P>0$ a constant independent of $n$. 
In Appendix \ref{sec:free_ent_conc} we verify by standard arguments that the 
following simple condition implies \eqref{concpressuregeneral} (we do not need $J_X \geq 0$ for this implication, see Proposition \ref{thm:free_ent_conc}):
\begin{condition} \label{conditionC}
We assume that $J_X$, $X\subset \{1, \dots, n\}$ are independent random variables with finite second moment and such that
$\sum_{X\subset \{1, \dots, n\}} {\rm Var}(J_X) \leq C_P n$ for a numerical constant $C_P>0$ independent of $n$.
\end{condition}
Models of physical interest also have a well-defined thermodynamic limit for $p_n$. This requires a little bit 
more structure on the distribution of the couplings $J_X$ and will not be used (see \cite{bayati2010combinatorial,salez2016interpolation} for proofs of the existence of such limits).
For completeness we give a simple hypothesis and standard argument in Appendix \ref{sec:free_ent_conc} that 
guarantees the existence of the thermodynamic limit in the case of ferromagnetic models (see Proposition \ref{propD2}). 
 
Let us give a canonical example of ferromagnetic spin system on a dilute random graph where Condition~\ref{conditionC} is satisfied as well as the existence of the thermodynamic limit of the pressure. 
Note that our results also cover dense graph systems as long as $J_X$ are suitably rescaled with $n$ so that Condition~\ref{conditionC} is met. 

\begin{example}[$p$-spin models on the Erdoes-R\'enyi hypergraph]\label{ex1}
A {\it dilute} ferromagnetic $p$-spin model with coupling strength $J >0$ and magnetic field $H>0$ can be constructed as follows. For all subsets 
$X\subset\{1,\dots, n\}$ with cardinalities different from $1$ and $p$ set $J_X =0$.
For all $X$ such that $\vert X\vert =1$ set $J_X = H$. In other words the Hamiltonian contains the one-body term $- H\sum_{i=1}^n \sigma_i$.  For all subsets with 
$\vert X\vert =p$ take for $J_X$ Bernoulli random variables with $\mathbb{P}(J_X=J) = \gamma n/\binom{n}{p}$ and $\mathbb{P}(J_X=0) = 1 - \gamma n/\binom{n}{p}$ where $\gamma >0$. 
The Hamiltonian thus contains on average of the order of $\gamma n$ interaction terms of the form $-J \sigma_{{a_1}} \sigma_{{a_2}} \ldots \sigma_{{a_p}}$ where 
 $a_i$, $i=1, \dots, p$ are chosen uniformly at random in $\{1, \dots, n\}$ without repetition.

This model can be generalized to mixed $p$-spin models as follows: Fix $H>0$, $J_2, \dots, J_{p^*} >0$, $\gamma_2, \dots, \gamma_{p^*} >0$. Let $k=2,\dots, p^*$. We then draw Bernoulli random variables for the 
couplings of subsets with cardinality $\vert X\vert = k$ such that $\mathbb{P}(J_X = J_k) = \gamma_k n/\binom{n}{k}$ and $\mathbb{P}(J_X=0) = 1 - \gamma_k n/\binom{n}{k}$. And 
again of course $J_X = H$ for $X$ such that $\vert X\vert =1$ and $J_X =0$ if instead $\vert X\vert \neq 1, 2, \dots, p^*$. These models are generalizations of the Ising two-body ferromagnet on a standard 
Erdoes-R\'enyi random graph. 
In  Appendix \ref{sec:free_ent_conc} we verify that Condition~\ref{conditionC} is satisfied so that $P_n$ concentrates on $p_n$, and also that the thermodynamic limit of $p_n$ exists.
\end{example}



Any observable is a linear combination over subsets $T\subset\{1,\dots, n\}$ of $\sigma_T \equiv \prod_{i\in T}\sigma_i$ and their Gibbs expectation 
is denoted by 
\begin{align*}
 \langle \sigma_T\rangle \equiv \frac{1}{\cal Z}\sum_{\bsig\in\{\pm1\}^n}\sigma_T \exp(-\sH_{0}(\bsig))\,.
\end{align*}
The crucial property of ferromagnetic models that will be instrumental in our analysis are the Griffiths-Kelly-Sherman (GKS) correlation inequalities. In their most general form these state that
\begin{align}
\< \sigma_T \>\geq 0 \qquad \text{and} \qquad \< \sigma_T \sigma_S \> - \< \sigma_T \> \< \sigma_S \> \geq 0\, \label{eq:GKS}
\end{align}
for any subsets $T, S \subset\{1,\dots, n\}$ and any set of coupling constants $J_X\geq 0$, $X\subset \{1,\dots, n\}$.

The {\it multi-overlaps} (simply called overlaps) are defined for any integer $k\geq 1$ as 
\begin{align}
 Q_k \equiv \frac{1}{n}\sum_{i=1}^n \sigma_i^{(1)}\cdots \sigma^{(k)}
 \label{def:overlap}
\end{align}
where $\{\bsig^{(\alpha)}, \alpha=1, \dots, k\}$ is a set of $k$ {\it replicas} of the spin configurations drawn according to the $k$-fold tensor product of the Gibbs measure. We emphasize that 
in this work
the replicas are always uncoupled and i.i.d.. The Gibbs average w.r.t. the tensor product Gibbs measure is still indicated as $\langle - \rangle$. 
Note that 
\begin{align*}
 \langle Q_k\rangle = \frac{1}{n}\sum_{i=1}^n \langle \sigma_i^{(1)}\cdots \sigma^{(k)}\rangle = \frac{1}{n}\sum_{i=1}^n \langle \sigma_i\rangle^k, \qquad k\geq 1\, .
\end{align*}

It is well known that concentration results for overlaps generally require the addition of small perturbation terms 
whose role is to select ``pure states.'' With suitable such perturbations, and under a suitable concentration hypothesis for the pressure (of the perturbed model) 
we show that for large $n$: {\it i) For any instance of the random model (i.e. of the quenched disorder) $Q_k$ concentrates over $\langle Q_k\rangle$; 
and ii) 
$\langle Q_k\rangle$ concentrates over $\mathbb{E}\langle Q_k\rangle$.} 
These two concentration properties imply that overall $Q_k$ concentrates on $\mathbb{E}\langle Q_k\rangle$. 
One then expects that the pressure is given by the cavity (or replica symmetric) formula but this is still an open problem. 

We consider one-body perturbation terms $\sH_{\mathrm{pert}}(\bsig)$ added to the generic Hamiltonian \eqref{mostgeneralhamiltonian}:
\begin{align}
\sH_{\mathrm{pert}}(\bsig) \equiv - h_0 \sum_{i=1}^{n} \sigma_i - h_1 \sum_{i=1}^n \tau_i \sigma_i
\label{pert}
\end{align}
with $h_0\in [0,1]$, $h_1\in [0,1]$, $\tau_i \sim \mathrm{Poi}(\alpha n^{\theta-1})$ i.i.d. for $i=1,\ldots,n$ and $\alpha\in [0,1]$, $\theta \in \big (1/2, 7/8]$.
The first part of the perturbation, proportional to $h_0$, is called {\it homogeneous perturbation} while the second part proportional to $h_1$ is called
{\it Poisson perturbation}. Both are purely ferromagnetic such that the GKS inequalities remain valid. Note that in 
distribution $h_1 \sum_{i=1}^n \tau_i \sigma_i \stackrel{d}{=} h_1 \sum_{v=1}^{\Gamma} \sigma_{i_v}$
where $\Gamma \sim \mathrm{Poi}(\alpha n^{\theta})$ and $i_v$ is randomly and uniformly chosen from $\{1,\dots,n\}$.
While this second expression might seem more natural, the first equivalent expression allows a more compact notation in our analysis.
We associate to the total Hamiltonian \eqref{mostgeneralhamiltonian} + \eqref{pert}, i.e.
\begin{align}
\sH(\bsig)\equiv\sH_0(\bsig)	+\sH_{\mathrm{pert}}(\bsig)\,,\label{Htot}
\end{align}
 its 
partition function ${\cal Z}_{h_0, h_1, \alpha}$, Gibbs expectation $\<-\>_{h_0, h_1, \alpha}$, pressure $P_n(h_0, h_1, \alpha)\equiv n^{-1}\ln {\cal Z}_{h_0, h_1, \alpha}$ and average pressure 
$p_n(h_0, h_1, \alpha) \equiv \mathbb{E}\,P_n(h_0, h_1, \alpha)$ defined similarly as before with $\sH_0$ replaced by $\sH$. Here $\mathbb{E}$ is the expectation over all quenched variables, i.e., $J_X$, $X\subset\{1, \dots, n\}$ and $\btau = (\tau_1, \dots, \tau_n)$. An elementary argument shows that the perturbation does not 
modify the thermodynamic properties as long as $h_0 \to 0_+$ ($\alpha$ can be taken fixed). More precisely in Appendix \ref{proof3.3} we show
\begin{align}
 \vert p_n(h_0, h_1, \alpha) - p_n(0,0,0) \vert \leq h_0 + \frac{\alpha h_1}{n^{1-\theta}} \,.
 \label{approxinequ}
\end{align}
In particular $\lim_{h_0\to 0_+}\lim_{n\to +\infty}\vert p_n(h_0, h_1, \alpha) - p_n\vert = 0$.

We can now state the main concentration results. From now on in the rest of the paper it is understood that $n$ is always large enough.
%
\begin{thm}[Thermal concentration of the overlaps]\label{thm:thermalconc}
Assume the Hamiltonian $\sH$ given by \eqref{Htot} satisfies $\< \sigma_i \sigma_j \> - \< \sigma_i \> \< \sigma_j \> \geq 0$ for all $i, j=1, \dots, n$. 
Then for any $[\underline{h}, \bar h]\subset (0,1)$, $h_1 \in [0,1]$, $\alpha\in [0,1]$, we have for any instance of the random Hamiltonian
\begin{align*}
\int_{\underline{h}}^{\bar h} dh_0 \,\big\< \big(Q_k - \< Q_k \>_{h_0, h_1, \alpha}\big)^2 \big\>_{h_0, h_1, \alpha}
	\leq \frac{2k}{n}\,.
\end{align*}
\end{thm}

The next two results use the concentration of the pressure \eqref{concpressuregeneral} for the total Hamiltonian \eqref{Htot}. 
This holds under Condition~\ref{conditionC} and the constant $C_P$ can easily be made independent of 
$h_0, h_1, \alpha$. 

\begin{thm}[Total concentration of the magnetization/first overlap]\label{thm:firstoverlap}
Assume the Hamiltonian $\sH$ given by \eqref{Htot} satisfies $\< \sigma_i \sigma_j \> - \< \sigma_i \> \< \sigma_j \> \geq 0$ for all $i, j=1, \dots, n$ and assume also that 
Condition~\ref{conditionC} holds.
Then for any $[\underline{h},\bar h]\subset (0,1)$, $h_1\in [0,1]$, $\alpha\in [0,1]$ we have
\begin{align*}
\int_{\underline{h}}^{\bar h} dh_0 \,\mathbb{E}\big\< \big(Q_1 - \mathbb{E}\< Q_1 \>_{h_0, h_1, \alpha}\big)^2 \big\>_{h_0, h_1, \alpha}
	\leq \frac{5C_P + 42}{n^{1/3}}\,.
\end{align*}
\end{thm}
\noindent {\bf Remark:} Let us make a few remarks about these two theorems. First of all the Poisson perturbation is not needed and we can set $h_1=\alpha=0$. Second, when both GKS inequalities hold $2k/n$ can be replaced by $k/n$ in Theorem \ref{thm:thermalconc} as it will become clear from the proof.
More interestingly: Assuming only $\< \sigma_i \sigma_j \> - \< \sigma_i \> \< \sigma_j \> \geq 0$ is weaker than assuming both GKS inequalities and even weaker than assuming 
only the second one. This assumption is satisfied (for example) for all Hamiltonians 
satisfying the FKG inequality. An example is $J_X\geq 0$ except for one-body terms (magnetic fields) that may have any sign. So in particular, Theorems 
\ref{thm:thermalconc} and \ref{thm:firstoverlap} cover the random field Ising model (RFIM). Another example is strong enough ferromagnetic two-body terms and any sign for 
magnetic fields and higher order interactions. 

The next theorem assumes both GKS inequalities and its extension to systems satisfying only FKG is an open problem. It would be of interest to extend this theorem to the RFIM. 
Moreover both the homogeneous and Poisson perturbations play an important role in the proof. 

\begin{thm}[Total concentration of the overlaps]\label{thm:multiconc}
Assume the Hamiltonian $\sH$ given by \eqref{Htot} satisfies both GKS inequalities (in other words assume 
all $J_X\geq 0$) and also that Condition~\ref{conditionC} holds. Then for any $[\underline h, \bar h]\subset (0,1)$, $[\underline\alpha, \bar\alpha]\subset (0,1)$, $h_1\in (0,1]$ and $\theta\in(1/2,7/8]$ we have
\begin{align*}
\int_{\underline h}^{\bar h} dh_0 \int_{\underline \alpha}^{\bar \alpha} d\alpha \, 
\mathbb{E}\big\< \big(Q_k - \mathbb{E} \< Q_k \>_{h_0, h_1, \alpha} \big)^2\big\>_{h_0, h_1, \alpha}
\leq\frac{4k^2}{(\tanh h_1)^k}\frac{\sqrt{5C_P+40}}{n^{(\theta-1/2)/3}} \,.
\end{align*}
\end{thm}

The bound yileds a decay $\mathcal{O}(n^{-1/8})$ for $\theta=7/8$. Also note that the prefactor on the r.h.s. grows exponentially fast with $k$. The details of the proof
show that a slowly growing $h_1$ can be accommodated and we can take $h_1= {\cal O}(\ln n)$ to mitigate this growth.


%
%
\section{Proofs of concentrations for the overlaps}\label{sec:conc}
The main aim of this section is to prove Theorem \ref{thm:multiconc}. The proof is generic and essentially requires only 
two ingredients: $i)$ That the Hamiltonian $\sH$ given by \eqref{Htot} is purely ferromagnetic so that the two GKS inequalities \eqref{eq:GKS} are verified;
$ii)$ the pressure of the perturbed model concentrates in the sense of \eqref{concpressuregeneral}. In the process we 
also obtain the proofs of Theorems \ref{thm:thermalconc} and \ref{thm:firstoverlap}. 

To ease the notations in this section we do not indicate explicitly the arguments $h_0$, $h_1$, $\alpha$ in the Gibbs brackets and pressure.

%
\subsection{Preliminary remarks}

Theorem \ref{thm:multiconc} will be a consequence of the individual control of three types of overlap fluctuations. 
One can verify by expanding the squares that the total overlaps fluctuations can be decomposed as
\begin{align}
\mathbb{E}  \big\< (Q_k - \mathbb{E} \< Q_k \> )^2 \big\>
	& = \mathbb{E} \big\< \big(Q_k - \< Q_k \> \big)^2 \big\>  + \mathbb{E} \big [ \big( \< Q_k \> - \mathbb{E}_{\btau} \< Q_k \> \big)^2 \big ] + \mathbb{E} \big [ \big( \mathbb{E}_{\btau} \< Q_k \> - \mathbb{E} \< Q_k \> \big)^2 \big ]\,.
	\label{eq:full-concentration-Qp:1}
\end{align}
The first type of fluctuations are purely {\it thermal fluctuations} and are controlled in Theorem \ref{thm:thermalconc} thanks to the homogeneous part of the perturbation in \eqref{pert}; for the analysis of these fluctuations the Poisson perturbation 
 can be dropped. The last two terms are the {\it disorder 
fluctuations} due to the quenched variables. Their control requires the Poisson perturbation\footnote{It is an open problem to assess if these can be dropped and the 
fluctuations controlled only thanks to the homogeneous perturbation.}. 
The second term are fluctuations directly related to the Poisson perturbation itself, called {\it Poisson fluctuations}, and is controlled by 
Lemma \ref{lemma4.4}. The third term are the fluctuations 
due to all other quenched couplings in the unperturbed Hamiltonian and is controlled by Lemma \ref{lemma4.8}.
In Section~\ref{sec:combiningConc} we show how to combine 
all these concentration results in order to obtain Theorem \ref{thm:multiconc}.

We will repeatedly make use of the following  technical lemma:
\begin{lemma}[A bound on differences of derivatives due to convexity]
Let $I\subset \mathbb{R}$ an open interval and $G: x\in I\mapsto G(x)\in \mathbb{R}$ and $g: x\in I\mapsto g(x)\in\mathbb{R}$ convex and differentiable functions. Fix $x\in I$ and let 
$\delta>0$ small enough so that $x\pm\delta\in I$.
Define $C^{+}_\delta(x) \equiv g'(x+\delta) - g'(x) \geq 0$ and $C^{-}_\delta(x) \equiv g'(x) - g'(x-\delta) \geq 0$. Then
\begin{align*}
|G'(x) - g'(x)| \leq \delta^{-1} \sum_{u \in \{x-\delta, x, x+\delta\}} |G(u)-g(u)| + C^{+}_\delta(x) + C^{-}_\delta(x)\,.
\end{align*}
\label{thm:bound-by-convexity}
\end{lemma}
\begin{proof}
Convexity implies that we have
\begin{align*}
G'(x) - g'(x)
 	& \leq \frac{G(x+\delta) - G(x)}{\delta} - g'(x) \nonumber \\
	& \leq \frac{G(x+\delta) - G(x)}{\delta} - g'(x) + g'(x+\delta) - \frac{g(x+\delta) - g(x)}{\delta} \nonumber \\
	& = \frac{G(x+\delta) - g(x+\delta)}{\delta} - \frac{G(x) - g(x)}{\delta} + C^+_\delta(x) \,,  \\
G'(x) - g'(x)
	& \geq \frac{G(x) - G(x-\delta)}{\delta} - g'(x) + g'(x-\delta) - \frac{g(x) - g(x-\delta)}{\delta} \nonumber \\
	& = \frac{G(x) - g(x)}{\delta} - \frac{G(x-\delta) - g(x-\delta)}{\delta} - C^-_{\delta}(x) \,.
\end{align*}
Combining these two inequalities ends the proof.
\end{proof}
\subsection{Thermal fluctuations of overlaps: Proof of Theorem \ref{thm:thermalconc}}
We start by considering the thermal fluctuations for a fixed realization of quenched variables. 
Note that 
\begin{align}
\frac{dP_n}{dh_0} = \frac{1}{n} \sum_{i=1}^{n} \< \sigma_i \> = \< Q_1 \> \,,\qquad 
\frac{1}{n} \frac{d^2 P_n}{dh_0^2} & = \frac{1}{n^2} \sum_{i,j=1}^{n} \left ( \< \sigma_i \sigma_j \> - \< \sigma_i \> \< \sigma_j \> \right ) = \big\< (Q_1 - \< Q_1 \> )^2 \big\>\,.\label{derivativesfF}
\end{align}
The second identity shows that the pressure $P_n$, as well as its expectation $p_n$, are convex in $h_0$ (a generic fact in statistical mechanics models). 

%
%
By the definition \eqref{def:overlap} of $Q_k$ we have
\begin{align}
\big\< (Q_k - \< Q_k \>)^2 \big\> 
	= \< Q_k^2 \> - \< Q_k \>^2  \-
	& = \frac{1}{n^2} \sum_{i,j=1}^{n} \left ( \< \sigma_i \sigma_j \>^k - \< \sigma_i \>^k \< \sigma_j \>^k \right ) \nonumber \\
	& = \frac{1}{n^2} \sum_{i,j=1}^{n}  ( \< \sigma_i \sigma_j \> - \< \sigma_i \> \< \sigma_j \> ) \sum_{l=0}^{k-1} \< \sigma_i \sigma_j \>^{k-1-l} \< \sigma_i \>^l \< \sigma_j \>^l \,.
	\label{eq:concentration1:1}
\end{align}
Using $\< \sigma_i \sigma_j \> - \< \sigma_i \> \< \sigma_j \> \geq 0$ for all $i, j=1, \ldots, n$ and the triangle inequality, \eqref{eq:concentration1:1} is upper bounded as
\begin{align*}
\big\< (Q_k - \< Q_k \>)^2 \big\>
	& \leq \frac{k}{n^2} \sum_{i,j=1}^{n}  ( \< \sigma_i \sigma_j \> - \< \sigma_i \> \< \sigma_j \> )\,.
\end{align*}
Hence integrating this inequality over $h_0 \in [\underline{h},\bar h]$ and using \eqref{derivativesfF} we have
\begin{align*}
\int_{\underline{h} }^{\bar h} dh_0 \,\big\< (Q_k - \< Q_k \>)^2 \big\>
	& \leq k \int_{\underline{h}}^{\bar h} dh_0 \,\frac{1}{n} \frac{d^2 P_n}{dh_0^2} \-
	= \frac{k}{n} \left [ \< Q_1 \> \right ]_{h_0=\underline{h}}^{h_0=\bar h} \-
	\leq \frac{2k}{n}\,.
\end{align*}
Note that if the first GKS inequality also holds then $\langle Q_1\rangle \geq 0$ so 
$0\leq \left [ \< Q_1 \> \right ]_{h_0=\underline{h}}^{h_0=\bar h}\leq 1$ and $2k/n$ becomes $k/n$.$\QED$
%
\subsection{Disorder fluctuations of the magnetization}
Before considering the concentration of general overlaps we need to control the quenched fluctuations of the first overlap $Q_1$, that is the magnetization. 
Indeed, our proof of the concentration of the overlaps w.r.t. the quenched variables in Section~\ref{sec:4.5} is 
based on an induction argument where the induction is on the order
$k$ of the overlaps $Q_k$, and the following lemma will serve as the base case for the induction. In order to control 
these fluctuations the homogeneous perturbation alone is again sufficient.
%
\begin{lemma}[Concentration of the magnetization w.r.t. the quenched disorder]\label{lemma4.3}
Assume that Condition~\ref{conditionC} holds. Then for any $[\underline{h},\bar h]\subset (0,1)$ we have
\begin{align*}
\int_{\underline h}^{\bar h} dh_0 \,\mathbb{E} \big [ (\<Q_1\> - \mathbb{E}\<Q_1\>)^2 \big ] \leq \frac{5C_P +40}{n^{1/3}}\,.
\end{align*}
\label{thm:concentration-Q1}
\end{lemma}
\noindent {\bf Remark:}
 There is no need to assume $J_X \geq 0$ here. So this lemma holds generally even if GKS or FKG inequalities do not hold.

\noindent {\bf Remark:}
 Combining Theorem \ref{thm:thermalconc} and Lemma \ref{lemma4.3} yields Theorem \ref{thm:firstoverlap}.
\begin{proof}
Below it will be convenient to indicate explicitly the $h_0$ dependence in the pressure. 
Recall $p_n(h_0)= \E\, P_n(h_0)$. From \eqref{derivativesfF} we have
\begin{align*}
\<Q_1\> - \mathbb{E}\<Q_1\> = \frac{dP_n(h_0)}{dh_0} - \frac{dp_n(h_0)}{dh_0}\,.
\end{align*}
Since $P_n$ and $p_n$ are convex in $h_0$ as seen from \eqref{derivativesfF}, we can use Lemma~\ref{thm:bound-by-convexity} to obtain
\begin{align*}
| \<Q_1\> - \mathbb{E}\<Q_1\> | \leq \delta^{-1} \sum_{u \in \{h_0-\delta, \delta, h_0+\delta\} } |P_n(u) - p_n(u)| + C^{+}_\delta(h_0) + C^{-}_\delta(h_0)
\end{align*}
where
\begin{align}
C^{+}_\delta(h_0) \equiv \frac{dp_n(h_0 + \delta)}{dh_0} - \frac{dp_n(h_0)}{dh_0} \geq 0\,,
\hspace{1cm} C^{-}_\delta(h_0) \equiv \frac{dp_n(h_0)}{dh_0} - \frac{dp_n(h_0 - \delta)}{dh_0} \geq 0\,.\label{C+-}
\end{align}
Squaring both sides, applying $ ( \sum_{r=1}^{k} u_r )^2 \leq k \sum_{r=1}^{k} u_r^2$, and then taking an expectation we have
\begin{align}
\mathbb{E} \big [ ( \<Q_1\> - \mathbb{E}\<Q_1\> )^2 \big ]
	& \leq 5 \delta^{-2}  \sum_{u \in \{h_0-\delta, \delta, h_0+\delta\} } \mathbb{E} [ (P_n(u) - p_n(u))^2 ] + 5 C^{+}_\delta(h_0)^2 + 5  C^{-}_\delta(h_0)^2\,.
	\label{eq:concentration-Q1:1}
\end{align}
Under the assumption \eqref{concpressuregeneral} about concentration of the pressure the first term is smaller than $5C_P/(n\delta^2)$. Next, using $|\frac{dp_n}{dh_0}|= |\E\<Q_1\>| \leq 1$ 
allows to assert from \eqref{C+-} the crude bound $C_\delta^{\pm}(h_0)\le 2$. Then using $C_\delta^{\pm}(h_0) \geq 0$,
\begin{align*}
\int_{\underline h}^{\bar h} dh_0 \big ( C^{+}_\delta(h_0)^2 + C^{-}_\delta(h_0)^2 \big )
	& \leq 2 \int_{\underline h}^{\bar h} dh_0 \big ( C^{+}_\delta(h_0) + C^{-}_\delta(h_0) \big ) 	
	\\
	& = 2 \big [ \big(p_n(\bar h+\delta) - p_n(\bar h-\delta)\big) + \big(p_n(\underline h-\delta) - p_n(\underline h + \delta)\big) \big ] \leq 8 \delta
\end{align*}
where the mean value theorem has been used to get the last inequality. When \eqref{eq:concentration-Q1:1} is integrated over $h_0$ we reach
\begin{align*}
\int_{\underline h}^{\bar h} dh_0\, \mathbb{E} \big [ (\<Q_1\> - \mathbb{E}\<Q_1\>)^2 \big ] \leq \frac{5C_P}{n\delta^2} + 40\delta \,.
\end{align*}
The proof is ended by optimizing the bound by choosing $\delta = n^{-1/3}$ for $n$ large enough. 
\end{proof}
\subsection{Poisson fluctuations of overlaps}
%
\begin{lemma}[Concentration of overlaps w.r.t. the Poisson perturbation]\label{lemma4.4}
Assume the Hamiltonian $\sH$ given by \eqref{Htot} is fully ferromagnetic so that both GKS inequalities hold. Then for any $[\underline h , \bar h]\subset(0,1)$ we have
\begin{align*}
\int_{\underline h}^{\bar h} dh_0 \, \mathbb{E}_{\btau} \big[ ( \< Q_k \> - \mathbb{E}_{\btau} \< Q_k \>)^2 \big] \leq \frac{\alpha k h_1}{n^{1-\theta}}\,.
\end{align*}
\label{thm:concentration-Poisson}
\end{lemma}
\begin{proof}
Recall $\btau$ is a random vector with i.i.d. components $\tau_j\sim\mathrm{Poi}(\alpha n^{\theta-1})$ for $j \in \{1, \dots, n\}$. Let $\btau^j$ 
be the random vector that differs from $\btau$ {\it only} at the $j$--th component, which is replaced by a new $\tau'_j\sim\mathrm{Poi}(\alpha n^{\theta-1})$ drawn 
independently from everything else. For this proof we explicitly keep track of the $\btau$ dependence in Gibbs expectations $\<-\>_\btau$. The Efron--Stein inequality states ($\mathds{1}(\cdot)$ is the indicator function) 
\begin{align}
\mathbb{E}_{\btau} \big[ &( \< Q_k \>_{\btau} - \mathbb{E}_{\btau} \< Q_k \>_{\btau} )^2 \big ] 
	 \leq \frac{1}{2} \sum_{j=1}^{n} \mathbb{E}_{\btau \setminus \tau_j}\mathbb{E}_{\tau_j} \mathbb{E}_{\tau'_j} \big[ ( \< Q_k \>_{\btau^{j}} - \< Q_k \>_{\btau} )^2 \big ] 
	 \nonumber \\
	& = \frac{1}{2}\sum_{j=1}^{n} \mathbb{E}_{\btau\setminus \tau_j} \mathbb{E}_{\tau_j}  
	\mathbb{E}_{\tau'_j} \big [ \mathds{1}(\tau'_j > \tau_j) ( \< Q_k \>_{\btau^{j}} - \< Q_k \>_{\btau} )^2 \big ]
	+
	\frac{1}{2}\sum_{j=1}^{n} \mathbb{E}_{\btau\setminus \tau_j} \mathbb{E}_{\tau_j}  
	\mathbb{E}_{\tau'_j} \big [ \mathds{1}(\tau'_j =\tau_j) ( \< Q_k \>_{\btau^{j}} - \< Q_k \>_{\btau} )^2 \big ]
	\nonumber \\ &
	\qquad + \frac{1}{2}\sum_{j=1}^{n} \mathbb{E}_{\btau\setminus \tau_j} \mathbb{E}_{\tau_j}
	\mathbb{E}_{\tau'_j} \big [ \mathds{1}(\tau'_j < \tau_j)( \< Q_k \>_{\btau^{j}} - \< Q_k \>_{\btau} )^2 \big ]
	\nonumber \\ &
	=
	\sum_{j=1}^{n} \mathbb{E}_{\btau\setminus \tau_j} \mathbb{E}_{\tau_j}  
	\mathbb{E}_{\tau'_j} \big [ \mathds{1}(\tau'_j > \tau_j) ( \< Q_k \>_{\btau^{j}} - \< Q_k \>_{\btau} )^2 \big ]\,.
	\label{eq:concentration-Poisson:1}
\end{align}
To get the second equality we used that the term with $\tau'_j =\tau_j$ vanishes and that the terms with $\tau'_j > \tau_j$ and $\tau'_j < \tau_j$ are equal by symmetry
(under exchange of $\tau_j^\prime$ and $\tau_j$). 
The two GKS inequalities 
imply 
\begin{align*}
 \frac{d \< Q_k \>_{\btau}}{d \tau_j} & = \frac{k h_1}{n}\sum_{i=1}^n \langle \sigma_i\rangle^{k-1} (\langle \sigma_i \sigma_j\rangle - \langle \sigma_i\rangle\langle \sigma_j\rangle)
 \geq 0
\end{align*}
and therefore $\< Q_k \>_{\btau^{j}} - \< Q_k \>_{\btau} \geq 0$ when $\tau'_j > \tau_j$ 
(here note that $\tau_j$ is an integer but we formally consider it real when computing a derivative and then restrict the obtained monotonicity result to the integer case). 
This together with $0\leq \< Q_k \>_{\btau}  \leq 1$ (by GKS) implies $( \< Q_k \>_{\btau^{j}} - \mathbb{E}_{\btau} \< Q_k \>_{\btau} ) \in [0,1]$. 
Then \eqref{eq:concentration-Poisson:1} implies
\begin{align}
\mathbb{E}_{\btau} \big[ ( \< Q_k \>_{\btau} - \mathbb{E}_{\btau} \< Q_k \>_{\btau} )^2 \big ] 
	& \leq \sum_{j=1}^{n} \mathbb{E}_{\btau\setminus \tau_j} \mathbb{E}_{\tau_j}  
	\mathbb{E}_{\tau'_j} \big [ \mathds{1}(\tau'_j > \tau_j) (\< Q_k \>_{\btau^{j}} - \< Q_k \>_{\btau} )\big ]\,.
	\label{eq:concentration-Poisson:2}
\end{align}
Let $\Delta_j \equiv \tau'_j - \tau_j$ and $\bu^j = (0, \dots, 0, 1, 0, \dots, 0)$ with $u_j = 1$. This allows us to rewrite $\btau^j = \btau + \Delta_j \bu^j$.
An interpolation gives
\begin{align}
	\< Q_k \>_{\btau + \Delta_j \bu^j} - \< Q_k \>_{\btau} \-
	&= \int_0^1 ds \frac{d}{ds} \< Q_k \>_{\btau + s \Delta_j \bu^j} \nonumber \\
	& = \frac{k h_1 \Delta_j}{n} \int_0^1 ds \sum_{i=1}^{n} \< \sigma_i\>_{\btau + s \Delta_j \bu^j}^{k-1} 
	\big (  \< \sigma_i \sigma_j \>_{\btau + s \Delta_j \bu^j} -  \< \sigma_i\>_{\btau + s \Delta_j \bu^j}  \< \sigma_j\>_{\btau + s \Delta_j \bu^j} \big )\,.
	\label{eq:concentration-Poisson:3}
\end{align}
Under the condition $\tau'_j > \tau_j$ the integrand in 
\eqref{eq:concentration-Poisson:3} is non-negative by the two GKS inequalities. Also note that, again by the two GKS inequalities,
\begin{align*}
\frac{d}{d \tau_j} \big ( \< \sigma_i \sigma_j \>_{\btau} -  \< \sigma_i\>_{\btau}  \< \sigma_j\>_{\btau} \big )
	& = - 2 h_1  \< \sigma_j\>_{\btau} \big ( \< \sigma_i \sigma_j \>_{\btau} -  \< \sigma_i\>_{\btau}  \< \sigma_j\>_{\btau} \big ) \leq 0
\end{align*}
so that, as $\Delta_j > 0$,
\begin{align}\label{eq:concentration-Poisson:5}
 \< \sigma_i \sigma_j \>_{\btau + s \Delta_j \bu^j} -  \< \sigma_i\>_{\btau + s \Delta_j \bu^j}  \< \sigma_j\>_{\btau + s \Delta_j \bu^j}
 \leq
 \< \sigma_i \sigma_j \>_{\btau} -  \< \sigma_i\>_{\btau}  \< \sigma_j\>_{\btau}\,.
\end{align}
Thus substituting \eqref{eq:concentration-Poisson:3}, \eqref{eq:concentration-Poisson:5} into \eqref{eq:concentration-Poisson:2} and simply upper bounding $\< \sigma_i\>_{\btau + s \Delta_j \bu^j}^{k-1}$ by $1$, we obtain
\begin{align}
\mathbb{E}_{\btau} \big [ ( \< Q_k \>_{\btau} - \mathbb{E}_{\btau} \< Q_k \>_{\btau} )^2 \big] 
	& \leq \frac{k h_1}{n} \sum_{i,j=1}^{n} \mathbb{E}_{\btau} \mathbb{E}_{\tau'_j} \big[ \mathds{1} ( \tau'_j > \tau_j ) (\tau'_j - \tau_j) \big (  \< \sigma_i \sigma_j \>_{\btau} -  \< \sigma_i\>_{\btau}  \< \sigma_j\>_{\btau} \big ) \big ]\,.
	\label{eq:concentration-Poisson:6}
\end{align}
The part containing $\tau'_j$ has an upper bound independent of $j$:
\begin{align*}
\mathbb{E}_{\tau'_j} \big[ \mathds{1} ( \tau'_j > \tau_j ) (\tau'_j - \tau_j) \big] 
	\leq \mathbb{E}_{\tau'_j} [\tau'_j]
	= \frac{\alpha}{n^{1-\theta}}
\end{align*}
because $\tau_j\ge 0$ and $\tau_j'\sim \mathrm{Poi}(\alpha n^{\theta-1})$. This further relaxes \eqref{eq:concentration-Poisson:6} to
\begin{align}
\mathbb{E}_{\btau} \big [ ( \< Q_k \>_{\btau} - \mathbb{E}_{\btau} \< Q_k \>_{\btau} )^2 \big ] 
	& \leq \frac{\alpha k h_1}{n^{2-\theta}} \sum_{i,j=1}^{n} \mathbb{E}_{\btau} \big [ \< \sigma_i \sigma_j \>_{\btau} -  \< \sigma_i\>_{\btau}  \< \sigma_j\>_{\btau} \big ]
	= \alpha k h_1 n^{\theta} \mathbb{E}_{\btau}  \big\< ( Q_1 - \< Q_1 \>_{\btau} )^2 \big\>_{\btau}  \label{eq:concentration-Poisson:8}
\end{align}
(recall \eqref{derivativesfF} for the last equality). 
Finally, integrating \eqref{eq:concentration-Poisson:8} over $h_0 \in [\underline h, \bar h]$ and using Theorem~\ref{thm:thermalconc} with $k=1$ (the factor $2$ can be 
removed here because we assume both GKS inequalities) ends the proof.
\end{proof}
\subsection{Last type of fluctuations of overlaps}\label{sec:4.5}
%
In this section we tackle the last kind of fluctuations in the decomposition \eqref{eq:full-concentration-Qp:1}. Before proceeding let us say a few words about the strategy. 
The proof is decomposed in three steps (where the first two follow the ideas used in proving the Ghirlanda-Guerra identities for spin glasses \cite{ghirlanda1998general}). 
The first step shows that the $\alpha$-derivative of the pressure concentrates 
which will lead to Lemma \ref{lemma4.5} (recall $\alpha$ controls the mean of the Poisson quenched variables $\tau_i$). In the second step we derive an identity which links a ``generating series'' containing overlap covariances to the product of an overlap 
and the pressure $\alpha$-derivative fluctuations. Using the concentration result of step one we can
 show that this generating series concentrates, leading to Lemma \ref{lemma4.6}. 
 In the third step, from the concentration of this generating series we extract the concentration of {\it each} overlap covariance, leading to Lemma \ref{lemma4.8}. 
 In particular, this will imply the control of the third kind of fluctuations in \eqref{eq:full-concentration-Qp:1}. 
 The third step is non-trivial as the generating series has alternating signs. 
 Nevertheless, we overcome this problem using an induction argument over $k$ (the order of the overlap) thanks to the GKS inequalities and to 
 Lemma~\ref{lemma4.3} for the base case $k=1$. 
\subsubsection{Step 1: Concentration of the pressure $\alpha$-derivative}
Let $\hat{P}_n(\alpha) \equiv \mathbb{E}_{\btau}P_n$. Note that the pressure $p_n(\alpha) = \E \,\hat P_n$ is obtained by taking an expectation over the rest of the quenched variables. 
We start with a few preliminaries about these functions.
We emphasize the $\alpha$ dependence in this section. As before let $\bu^j = (0, \dots, 0, 1, 0, \dots, 0)$ with $u_j = 1$. Recall $\tau_i\sim \mathrm{Poi}(\alpha n^{\theta-1})$. Then a straightforward algebra using the 
Poisson property \eqref{eq:dPoissonMean} yields the following identities:
\begin{align}
\frac{d\hat{P}_n(\alpha)}{d\alpha} 
	& = \frac{1}{n^{2-\theta}} \sum_{i=1}^n \mathbb{E}_{\btau} \ln \< e^{h_1\sigma_i } \>_{\btau} \nonumber\\
	&= \frac{1}{n^{2-\theta}} \sum_{i=1}^n \mathbb{E}_{\btau} \ln ( 1 + \< \sigma_i \>_{\btau} \tanh h_1 ) + 
	\frac{1}{n^{1-\theta}} \ln\cosh h_1
	\,, \label{eq:Ft-der-alpha} \\
\frac{d^2 \hat{P}_n(\alpha)}{d\alpha^2} 
	& = \frac{1}{n^{3-2\theta}} \sum_{i, j=1}^{n} \Big( \mathbb{E}_{\btau} \ln \< e^{h_1\sigma_i } \>_{\btau+\bu^{j}} - \mathbb{E}_{\btau} \ln \< e^{h_1\sigma_i } \>_{\btau} \Big) \nonumber \\
	& = \frac{1}{n^{3-2\theta}} \sum_{i,j=1}^{n} \Big (  \mathbb{E}_{\btau} \ln ( 1 + \< \sigma_i \>_{\btau+\bu^j} \tanh h_1 ) -  \mathbb{E}_{\btau} \ln ( 1 + \< \sigma_i \>_{\btau} \tanh h_1 ) \Big)\,, \label{eq:Ft-der2-alpha}
\end{align}
where we used $e^{\sigma x} = \cosh x (1 +\sigma \tanh x)$ for $\sigma=\pm1$. The derivatives for $p_n(\alpha)$ can directly 
be obtained by taking an expectation over the rest of the quenched variables:
\begin{align}
\frac{dp_n(\alpha)}{d\alpha} 
	& =\frac{1}{n^{2-\theta}} \sum_{i=1}^n \mathbb{E} \ln ( 1 + \< \sigma_i \>_{\btau} \tanh h_1 ) + \frac{1}{n^{1-\theta}}\ln\cosh h_1
	\,, \label{eq:f-der-alpha} \\
\frac{d^2 p_n(\alpha)}{d\alpha^2} 
	& = \frac{1}{n^{3-2\theta}} \sum_{i,j=1}^n  \Big (\mathbb{E} \ln ( 1 + \< \sigma_i \>_{\btau+\bu^j} \tanh h_1 ) - \mathbb{E} \ln ( 1 + \< \sigma_i \>_{\btau} \tanh h_1 ) \Big )\,. \label{eq:f-der2-alpha}
\end{align}
The second GKS inequality \eqref{eq:GKS} implies that $d \< \sigma_i \>_{\btau}/d \tau_j = h_1(\< \sigma_i\sigma_j \>_{\btau}-\< \sigma_i \>_{\btau}\< \sigma_j \>_{\btau})$ is non-negative, and therefore $ \< \sigma_i \>_{\btau} \leq \< \sigma_i \>_{\btau+\bu^j}$. 
Thus the identities \eqref{eq:Ft-der2-alpha} and \eqref{eq:f-der2-alpha} imply (using also $1+\langle \sigma_i\rangle \tanh h_1 \ge 0$) that $d^2 \hat{P}_n(\alpha)/d\alpha^2 \geq 0$ and $d^2 p_n(\alpha)/d\alpha^2 \geq 0$, 
which means that $\hat{P}_n(\alpha)$ and $p_n(\alpha)$ are convex in $\alpha$ (note that in order to obtain this convexity we used only the second GKS inequality, no need of the first one here).
One can also see that 
\begin{align}
\Big|\frac{dp_n(\alpha)}{d\alpha} \Big|\leq \frac{h_1}{n^{1-\theta}}\label{68}	\,.
\end{align}

We can now show a concentration result for the $\alpha$-derivative of the pressure based on Lemma \ref{thm:bound-by-convexity}.
\\
\begin{lemma}[Concentration of the pressure $\alpha$-derivative]\label{lemma4.5}
Assume the Hamiltonian $\sH$ given by \eqref{Htot} satisfies $\< \sigma_i \sigma_j \> - \< \sigma_i \> \< \sigma_j \> \geq 0$ for all $i, j=1, \dots, n$ and assume also that 
Condition~\ref{conditionC} holds. Then for any $[\underline \alpha,\bar \alpha]\subset(0,1)$, $h_0\in(0,1)$, $h_1\in(0,1]$ and $\theta \in(1/2,1)$ we have
\begin{align*}
\int_{\underline \alpha}^{\bar \alpha} d\alpha \, \mathbb{E} \Big [ \Big ( \frac{d\hat{P}_n(\alpha)}{d \alpha} - \frac{dp_n(\alpha)}{d\alpha} \Big )^2 \Big ]
	& \leq \frac{5C_P+40h_1^2}{n^{(5-4\theta)/3}}\,. 
\end{align*}
\label{thm:concentration2a}
\end{lemma}
\begin{proof}
The proof is similar to the one of Lemma \ref{lemma4.3}. Using Lemma~\ref{thm:bound-by-convexity}, the convexity of $\hat{P}_n$ and $p_n$ in $\alpha$ implies that for any $\delta > 0$ we have
\begin{align*}
\Big| \frac{d\hat{P}_n(\alpha)}{d\alpha} - \frac{dp_n(\alpha)}{d\alpha} \Big |
	& \leq  \delta^{-1} \sum_{\substack{u \in \{ \alpha - \delta, \alpha, \alpha + \delta \}}} | \hat{P}_n(u) - p_n(u) |  + C_\delta^{+}(\alpha) + C_\delta^{-}(\alpha)	
\end{align*}
(note we can take $\delta$ small enough so that $\alpha -\delta>0$) where
\begin{align*}
C_\delta^{+}(\alpha) \equiv \frac{dp_n(\alpha + \delta)}{d\alpha} - \frac{dp_n(\alpha)}{d\alpha} \geq 0\,,\qquad C_\delta^{-}(\alpha) \equiv\frac{dp_n(\alpha)}{d\alpha} - \frac{dp_n(\alpha - \delta)}{d\alpha} \geq 0\,.
\end{align*}
Squaring both sides, applying $( \sum_{r=1}^{k} u_r )^2 \leq k \sum_{r=1}^{k} u_r^2$ and averaging we get
\begin{align}
\mathbb{E} \Big [ \Big ( \frac{d\hat{P}_n(\alpha)}{d\alpha} - \frac{dp_n(\alpha)}{d\alpha} \Big )^2 \Big ]
	& \leq 5 \delta^{-2} \sum_{\substack{u \in \{ \alpha - \delta, \alpha, \alpha + \delta \}}} \mathbb{E} [ ( \hat{P}_n(u) - p_n(u)  )^2  ] 
	+ 5C_\delta^{+}(\alpha)^2 + 5C_\delta^{-}(\alpha)^2\,.
	\label{eq:concentration2a:bd1}
\end{align}
It is easy to check that 
\begin{align*}
 \mathbb{E}[(\hat{P}_n(\alpha)-p_n(\alpha))^2] = \mathbb{E}[(P_n(\alpha) - p_n(\alpha))^2] - \mathbb{E}[(P_n(\alpha) - \hat{P}_n(\alpha))^2] \leq \mathbb{E}[(P_n(\alpha) - p_n(\alpha))^2]\,.
\end{align*}
Thus under the concentration assumption \eqref{concpressuregeneral} for the pressure, the first term in the r.h.s. of \eqref{eq:concentration2a:bd1} is smaller than $5C_P/(n\delta^2)$. 
Next, we recall \eqref{68} which implies the crude bound $C_\delta^{\pm}(h_0)\le 2h_1/ n^{1-\theta}$, so using $C_\delta^{\pm}(h_0) \geq 0$,
\begin{align*}
\int_{\underline \alpha}^{\bar \alpha} d\alpha \big ( C_\delta^{+}(\alpha)^2 + C_\delta^{-}(\alpha)^2 \big )
	& \leq \frac{2h_1}{n^{1-\theta}} \int_{\underline \alpha}^{\bar \alpha} d\alpha \big ( C_\delta^{+}(\alpha) + C_\delta^{-}(\alpha) \big ) \\
	& = \frac{2h_1}{n^{1-\theta}} \big[ \big ( p_n(\bar \alpha+\delta) - p_n(\bar \alpha-\delta) \big ) + \big ( p_n(\underline{\alpha}-\delta) - p_n(\underline{\alpha}+\delta) \big )  \big ] \leq \frac{8 \delta h_1^2}{n^{2- 2\theta}}
\end{align*}
where we used the mean value theorem for the last inequality. Thus when \eqref{eq:concentration2a:bd1} is integrated over $\alpha$ we obtain
\begin{align}\label{lastbounddd}
\int_{\underline \alpha}^{\bar \alpha} d\alpha \, \mathbb{E} \Big [ \Big ( \frac{d\hat{P}_n(\alpha)}{d\alpha} - \frac{dp_n(\alpha)}{d\alpha} \Big )^2 \Big ]
	& \leq \frac{5C_P}{n \delta^2} + \frac{40 \delta h_1^2}{n^{2- 2\theta}}\,.
\end{align}
The proof is ended by choosing $\delta$ such that $n^{-1}\delta^{-2} = \delta n^{-2+2\theta}$, in other words $\delta= n^{(1-2\theta)/3}$, which is possible for $\theta > 1/2$ (because we must have 
$\delta$ small enough in \eqref{lastbounddd}). 
With this choice the upper bound in \eqref{lastbounddd} becomes $(5C_P + 40 h_1^2)/n^{(5-4\theta)/3}$. Note that $5-4\theta >0$ because $\theta <1$ anyway. 
\end{proof}
\subsubsection{Step 2: Linking the fluctuations of the pressure $\alpha$-derivative to a series of overlap covariances}
In this step $K\geq 1$ is an integer fixed throughout. Define the set of multi-overlap covariances 
(w.r.t. the quenched variables except the Poisson ones $\btau$) as
\begin{align}
{\rm Cov}_{{K}, k} \equiv \mathbb{E} [ \mathbb{E}_{\btau} \< Q_{{K}} \> \, \mathbb{E}_{\btau} \< Q_k \> ] - \mathbb{E} \< Q_{K} \>  \,\mathbb{E} \< Q_k \>\,, \qquad k\ge 1\,.\label{cov}
\end{align}
The task is to bound the variance of 
$\mathbb{E}_{\btau} \< Q_k \>$ using Lemma \ref{lemma4.5}. However, here 
is a case where constructing a bound for the covariances is more flexible and feasible. 
Roughly speaking, we will show in this step that a {\it generating series} for the set $\{{\rm Cov}_{K, k}, k\geq 1\}$ is small. 
From this knowledge, and despite this series has alternating signs, we will in step 3 deduce that all individual covariances ${\rm Cov}_{K, k}$ 
are also small. In particular this will hold for the variance term $k=K$. 

\begin{lemma}[Concentration of a generating series]\label{lemma4.6}
Assume the Hamiltonian $\sH$ given by \eqref{Htot} satisfies $\< \sigma_i \sigma_j \> - \< \sigma_i \> \< \sigma_j \> \geq 0$ for all $i, j=1, \dots, n$ and assume also that 
Condition~\ref{conditionC} holds. Then for any $[\underline \alpha,\bar \alpha]\subset(0,1)$, $h_0\in(0,1)$, $h_1\in(0,1]$, $\theta \in(1/2,1)$ and any fixed integer $K\geq 1$ we have
\begin{align}
\int_{\underline{\alpha}}^{\bar \alpha} d\alpha \, \Big | \sum_{k=1}^{\infty} \frac{(-1)^{k+1}}{k} (\tanh h_1)^k\, {\rm Cov}_{K, k} \Big |
\leq 
\frac{\sqrt{5C_P+40h_1^2}}{n^{(\theta-1/2)/3}}\,.
\label{eq:concentration-Qp:sum}
\end{align}
\label{thm:concentration-Qp:sum}
\end{lemma}
\begin{proof}
By the Cauchy-Schwarz inequality and Lemma~\ref{thm:concentration2a} we have
\begin{align}
\int_{\underline \alpha}^{\bar \alpha}d\alpha \, & \Big | \mathbb{E} \Big [ \mathbb{E}_{\btau} \< Q_{K} \>  \Big ( \frac{d\hat{P}_n(\alpha)}{d\alpha} - \frac{dp_n(\alpha)}{d\alpha} \Big ) \Big ] \Big | 
\nonumber \\
	& \leq  \Big \{ \int_{\underline \alpha}^{\bar \alpha}d\alpha \, \mathbb{E} \big [ \big ( \mathbb{E}_{\btau} \< Q_{K} \> \big )^2 \big ] \Big \}^{1/2} \Big \{ \int_{\underline \alpha}^{\bar \alpha}d\alpha \, \mathbb{E} \Big [ \Big ( \frac{d\hat{P}_n(\alpha)}{d\alpha} - \frac{dp_n(\alpha)}{d\alpha} \Big )^2 \Big ] \Big \}^{1/2} \nonumber \\
	& \leq \frac{\sqrt{5C_P+40h_1^2}}{n^{(5-4\theta)/6}}\,. \label{eq:concentration2b}
\end{align}
The next step is to expand the $\alpha$-derivatives of the pressure. For that we recall the formulas \eqref{eq:Ft-der-alpha} and \eqref{eq:f-der-alpha}. Taylor 
expanding the logarithms in \eqref{eq:Ft-der-alpha} and recalling $n^{-1}\sum_{i=1}^n\< \sigma_i \>^k = \< Q_k \>$ gives
\begin{align*}
\frac{d\hat{P}_n(\alpha)}{d\alpha} 
	& = \frac{1}{n^{1-\theta}} \sum_{k=1}^{\infty} \frac{(-1)^{k+1}}{k} (\tanh h)^k \,\mathbb{E}_{\btau}  \< Q_k \> + \frac{1}{n^{1-\theta}} \ln \cosh h_1\,. 
\end{align*}
The series expansion of $\frac{dp_n}{d\alpha}$ is obtained similarly based on \eqref{eq:f-der-alpha}, and is thus the same with $\E_\btau$ replaced by 
the full expectation $\E$. Substituting these expansions in the left-most hand side of \eqref{eq:concentration2b} yields 
\begin{align*}
\int_{\underline \alpha}^{\bar \alpha}d\alpha \, \Big | \sum_{k=1}^{\infty} \frac{(-1)^{k+1}}{k} (\tanh h_1)^k \big\{  \mathbb{E} [ \mathbb{E}_{\btau} \< Q_{K} \>\, \mathbb{E}_{\btau} \< Q_k \> ] - \mathbb{E} \< Q_{K}\> \, \mathbb{E}  \< Q_k \>\big\} \Big |
	& \leq \frac{\sqrt{5C_P+40h_1^2}}{n^{(5-4\theta)/6}}n^{1-\theta}\,. 
\end{align*}
Recognizing \eqref{cov} then ends the proof.
\end{proof}
%
\subsubsection{Step 3: Induction argument over the overlap covariances}
We start with a useful monotonicity lemma that will allow us to control the alternating signs of the generating series in Lemma
\ref{lemma4.6}.
\begin{lemma}[A monotonicity property]
Assume the Hamiltonian $\sH$ given by \eqref{Htot} is fully ferromagnetic and thus satisfies both GKS inequalities. Then we have
\begin{align*}
\frac{1}{k}{\rm Cov}_{K,k} - \frac{\tanh h_1}{k+1} {\rm Cov}_{K,k+1} \geq 0\,.
\end{align*}
\label{thm:Harris:application}
\end{lemma}
\begin{proof}
Let $\bm{J}\equiv(J_X, X\subset\{1,\dots,n\})$. Define $g_{K}(\bm{J}) \equiv \mathbb{E}_{\btau} \<Q_{K} \>$ and $\tilde{g}_{k}(\bm{J}) \equiv \frac{1}{k} \mathbb{E}_{\btau} \<Q_{k}\> - \frac{\tanh h_1}{k+1} \mathbb{E}_{\btau} \< Q_{k+1} \>$. One can then recognize
\begin{align}
\frac{1}{k}{\rm Cov}_{K,k} - \frac{\tanh h_1}{k+1} {\rm Cov}_{K,k+1}
	& = \mathbb{E}[g_{K}(\bm{J}) \,\tilde{g}_{k}(\bm{J})] - \mathbb{E}\,g_{K}(\bm{J})\, \mathbb{E}\,\tilde{g}_{k}(\bm{J}) \label{eq:Harris:application:identify}
\end{align}
so it is enough to verify that $g_{K}(\bm{J})$ and $\tilde g_{k}(\bm{J})$ are positively correlated. Note that the expectations in \eqref{eq:Harris:application:identify} only carry over the set of
i.i.d. random coupling constants $J_X$, $X\subset\{1,\ldots,n\}$.
By the GKS inequalities the following derivatives are non-negative:
\begin{align*}
\frac{d}{dJ_X} g_{K}(\bm{J}) & = \frac{K}{n} \sum_{i=1}^{n} \mathbb{E}_{\btau} \big [ \< \sigma_i \>^{K-1} \big ( \< \sigma_i \sigma_X \> - \< \sigma_i \> \< \sigma_X \> \big ) \big ] \geq 0\,, \\
\frac{d}{dJ_X} \tilde{g}_{k}(\bm{J}) & = \frac{1}{n} \sum_{i=1}^{n} \mathbb{E}_{\btau} \big [ \< \sigma_i \>^{k-1} \big ( 1 - \< \sigma_i \> \tanh h_1 \big ) \big ( \< \sigma_i \sigma_X \> 
- \< \sigma_i \> \< \sigma_X \> \big ) \big ] \geq 0\,.
\end{align*}
In other words, $g_{K}(\bm{J})$ and $\tilde{g}_{k}(\bm{J})$ have same monotonicity w.r.t. each $J_X$ for all $X\subset\{1,\ldots,n\}$. 
We can then apply the Harris inequality (reproduced in Lemma~\ref{thm:Harris}, Appendix \ref{appendix-harris}) to finish the proof.
\end{proof}
Now we have all the necessary ingredients in order to inductively extract the concentration of each individual overlap from Lemma \ref{lemma4.6}.

\begin{lemma}[Concentration of the overlaps w.r.t. the quenched variables] \label{lemma4.8}
Assume the Hamiltonian $\sH$ given by \eqref{Htot} is fully ferromagnetic so that it satisfies both GKS inequalities, and also that Condition~\ref{conditionC} holds. 
Then for any $[\underline{h}, \bar h]\subset(0,1)$, $[\underline{\alpha},\bar\alpha]\subset(0,1)$, $\theta\in (1/2, 1)$, $h_1\in(0,1]$ and 
 any $k,K\ge 1$ we have
\begin{align*}
\int_{\underline h}^{\bar h} dh_0 \int_{\underline \alpha}^{\bar \alpha}d\alpha \, | {\rm Cov}_{K, k} | 
\leq \frac{kM_k}{(\tanh h_1)^k}\frac{\sqrt{5C_P+40}}{n^{(\theta-1/2)/3}}\,,
\end{align*}
where $M_k$ is defined by $M_1=1$, $M_{2k}=M_{2k-1}+1$, $M_{2k+1}=M_{2k}+2$ (so $M_k< 3k/2$). In particular for $k =K$,
\begin{align*}
\int_{\underline h}^{\bar h} dh_0 \int_{\underline \alpha}^{\bar \alpha} d\alpha\, \mathbb{E} \big [ \big ( \mathbb{E}_{\btau} \< Q_{k} \> 
- \mathbb{E}  \< Q_k \> \big )^2 \big ] \leq \frac{3k^2}{2(\tanh h_1)^k}\frac{\sqrt{5C_P+40}}{n^{(\theta-1/2)/3}}\,.
\end{align*}
\label{thm:Qp-induction}
\end{lemma}
\begin{proof}
We start the induction with the base case $k=1$. From \eqref{cov} we note that
\begin{align*}
 {\rm Cov}_{{K}, 1}  = \mathbb{E} [ \mathbb{E}_{\btau} &\< Q_{{K}} \> \, \mathbb{E}_{\btau} \< Q_1 \> ] - \mathbb{E} \< Q_{K} \>  \,\mathbb{E} \< Q_1 \>\nonumber\\
 &=\mathbb{E} [ \mathbb{E}_{\btau} \< Q_{{K}} \> \, \< Q_1 \> ] - \mathbb{E} \< Q_{K} \>  \,\mathbb{E} \< Q_1 \> =
 \mathbb{E} [ \mathbb{E}_{\btau} \< Q_{K} \> ( \< Q_1 \> -\mathbb{E} \< Q_1 \>)]\,.
\end{align*}
Then, using successively Fubini's theorem, the Cauchy-Schwarz inequality and Lemma~\ref{thm:concentration-Q1}, we have
\begin{align}
&\int_{\underline h}^{\bar h} dh_0 \int_{\underline \alpha}^{\bar \alpha} d\alpha \,| {\rm Cov}_{K, 1} | 
	 = \int_{\underline \alpha}^{\bar \alpha} d\alpha \int_{\underline h}^{\bar h} dh_0\, \big|\mathbb{E} [ \mathbb{E}_{\btau} \< Q_{K} \> (\< Q_1 \> -\mathbb{E} \< Q_1 \>)]\big| \nonumber\\
	& \qquad \qquad \leq \Big \{ \int_{\underline \alpha}^{\bar \alpha} d\alpha \int_{\underline h}^{\bar h} dh_0 \,\mathbb{E}  \big[  \big( \mathbb{E}_{\btau} \< Q_{K} \> \big)^2  \big] \Big \}^{1/2} \Big \{ \int_{\underline \alpha}^{\bar \alpha} d\alpha \int_{\underline h}^{\bar h} dh_0 \,\mathbb{E} \big[ \big( \< Q_1 \> - \mathbb{E}\< Q_1 \>  \big)^2  \big] \Big \}^{1/2} \nonumber \\
	&\qquad \qquad \leq \frac{\sqrt{5C_P+40}}{n^{1/6}} \leq \frac{\sqrt{5C_P+40}}{n^{(\theta-1/2)/3}\tanh h_1}\,.\label{CovK1}
\end{align}
Note that the last inequality is valid because $0<\tanh h_1\leq 1$ and $\theta < 1$. For $k \geq 2$ we adopt an induction in two steps: From $2k-1$ to $2k$ and then from $2k$ to $2k+1$. 

We start with the induction step from $2k-1$ to $2k$. Suppose 
\begin{align}
\int_{\underline h}^{\bar h} dh_0 \int_{\underline \alpha}^{\bar \alpha} d\alpha \,| {\rm Cov}_{K, 2k-1} | 
\leq \frac{(2k-1)M_{2k-1}}{(\tanh h_1)^{2k-1} }\frac{\sqrt{5C_P+40}}{n^{(\theta-1/2)/3}}\,. \label{eq:Qp-induction:case-I:1}
\end{align}
The left hand side of \eqref{eq:concentration-Qp:sum} is
\begin{align}
\int_{\underline \alpha}^{\bar \alpha}d\alpha   \, \Big | &\sum_{k'=1}^{\infty} \frac{(-1)^{k'+1}}{k'} (\tanh h_1)^{k'} {\rm Cov}_{K,k'} \Big |\nonumber\\
	&= \int_{\underline \alpha}^{\bar \alpha}d\alpha \, \Big | \sum_{k'=1}^{\infty} (\tanh h_1)^{2k'-1} \Big ( \frac{1}{2k'-1} {\rm Cov}_{K, 2k'-1} - \frac{\tanh h_1}{2k'} {\rm Cov}_{K, 2k'} \Big )  \Big | \nonumber \\
	&=\int_{\underline \alpha}^{\bar \alpha}d\alpha \, \sum_{k'=1}^{\infty} (\tanh h_1)^{2k'-1} \Big | \frac{1}{2k'-1} {\rm Cov}_{K, 2k'-1} - \frac{\tanh h_1}{2k'} {\rm Cov}_{K, 2k'} \Big | \label{eq:Qp-induction:case-I:2} \\
	&\geq\int_{\underline \alpha}^{\bar \alpha}d\alpha \, \Big | \frac{(\tanh h_1)^{2k-1}}{2k-1} {\rm Cov}_{K, 2k-1} - \frac{(\tanh h_1)^{2k}}{2k} {\rm Cov}_{K, 2k} \Big | 
	\label{eq:Qp-induction:case-I:3}
\end{align}
where \eqref{eq:Qp-induction:case-I:2} follows from $h_1\geq 0$ and Lemma~\ref{thm:Harris:application}. By the triangle inequality we have
\begin{align}
	&\frac{(\tanh h_1)^{2k}}{2k} \int_{\underline h}^{\bar h} dh_0 \int_{\underline \alpha}^{\bar \alpha}d\alpha \,  | {\rm Cov}_{K, 2k} | \nonumber \\
	& = \int_{\underline h}^{\bar h} dh_0  \int_{\underline \alpha}^{\bar \alpha}d\alpha \, \Big | \Big ( \frac{(\tanh h_1)^{2k-1}}{2k-1} {\rm Cov}_{K, 2k-1} - \frac{(\tanh h_1)^{2k}}{2k} {\rm Cov}_{K, 2k} \Big ) - \frac{(\tanh h_1)^{2k-1}}{2k-1} {\rm Cov}_{K, 2k-1} \Big | \nonumber \\
	& \leq \int_{\underline h}^{\bar h} dh_0 \int_{\underline \alpha}^{\bar \alpha}d\alpha \, \Big | \frac{(\tanh h_1)^{2k-1}}{2k-1} {\rm Cov}_{K, 2k-1} - \frac{(\tanh h_1)^{2k}}{2k} {\rm Cov}_{K, 2k} \Big | \nonumber\\
	&\qquad\qquad\qquad\qquad+ \frac{(\tanh h_1)^{2k-1}}{2k-1} \int_{\underline h}^{\bar h} dh_0 \int_{\underline \alpha}^{\bar \alpha}d\alpha \, | {\rm Cov}_{K, 2k-1}| \nonumber \\
	& \leq \int_{\underline h}^{\bar h} dh_0 \int_{\underline \alpha}^{\bar \alpha}d\alpha \, \Big | \sum_{k'=1}^{\infty} \frac{(-1)^{k'+1}}{k'} (\tanh h_1)^{k'} {\rm Cov}_{K, k'} \Big | + \frac{(\tanh h_1)^{2k-1}}{2k-1} \int_{\underline h}^{\bar h} dh_0 \int_{\underline \alpha}^{\bar \alpha}d\alpha \,| {\rm Cov}_{K, 2k-1}  | \label{eq:Qp-induction:case-I:4} \\
	& \leq
	\frac{\sqrt{5C_P+40h_1^2}}{n^{(\theta-1/2)/3}}+M_{2k-1}\frac{\sqrt{5C_P+40}}{n^{(\theta-1/2)/3}}\label{eq:Qp-induction:case-I:5}
	\\
	& \leq
	(M_{2k-1}+1) \frac{\sqrt{5C_P +40}}{n^{(\theta-1/2)/3}}
	\label{eq:Qp-induction:case-I:6}
\end{align}
where~\eqref{eq:Qp-induction:case-I:4} follows from~\eqref{eq:Qp-induction:case-I:3}, then \eqref{eq:Qp-induction:case-I:5} follows from Lemma~\ref{thm:concentration-Qp:sum} 
and the hypothesis~\eqref{eq:Qp-induction:case-I:1}, and finally \eqref{eq:Qp-induction:case-I:6} uses $h_1\in(0,1]$. Summarizing, we have shown
\begin{align}\label{endfirststep}
 \int_{\underline h}^{\bar h} dh_0 \int_{\underline \alpha}^{\bar \alpha}d\alpha \,  | {\rm Cov}_{K, 2k} | \leq\frac{2kM_{2k}}{(\tanh h_1)^{2k}}\frac{\sqrt{5C_P+40}}{n^{(\theta-1/2)/3}}
\end{align}
with $M_{2k}=M_{2k-1}+1$.

Now we proceed similarly for the induction from $2k$ to $2k+1$. 
This time we start with
\begin{align}
\int_{\underline \alpha}^{\bar \alpha} d\alpha\, \Big | &\sum_{k'=2}^{\infty} \frac{(-1)^{k'+1}}{k'} (\tanh h_1)^{k'} {\rm Cov}_{K, k'} \Big |\nonumber\\
	&= \int_{\underline \alpha}^{\bar \alpha} d\alpha \,\Big | \sum_{k'=1}^{\infty} (\tanh h_1)^{2k'} \Big ( \frac{1}{2k'} {\rm Cov}_{K, 2k'} - \frac{\tanh h_1}{2k'+1} {\rm Cov}_{K, 2k'+1} \Big )  \Big | \nonumber \\
	&=\int_{\underline \alpha}^{\bar \alpha} d\alpha \sum_{k'=1}^{\infty} (\tanh h_1)^{2k'} \Big | \frac{1}{2k'} {\rm Cov}_{K, 2k'} - \frac{\tanh h_1}{2k'+1} {\rm Cov}_{K, 2k'+1} \Big | \label{eq:Qp-induction:case-II:2} \\
	&\geq \int_{\underline \alpha}^{\bar \alpha} d\alpha \,\Big | \frac{(\tanh h_1)^{2k}}{2k} {\rm Cov}_{K, 2k} - \frac{(\tanh h_1)^{2k+1}}{2k+1} {\rm Cov}_{K, 2k+1} \Big | \label{eq:Qp-induction:case-II:3}
\end{align}
where~\eqref{eq:Qp-induction:case-II:2} follows from Lemma~\ref{thm:Harris:application} and $h_1\geq 0$. Also we have
\begin{align}
&\int_{\underline \alpha}^{\bar \alpha} d\alpha\, \Big | \sum_{k'=2}^{\infty} \frac{(-1)^{k'+1}}{k'} (\tanh h_1)^{k'} {\rm Cov}_{K, k'} \Big |
	 = \int_{\underline \alpha}^{\bar \alpha} d\alpha \,\Big | \sum_{k'=1}^{\infty}  \frac{(-1)^{k'+1}}{k'} (\tanh h_1)^{k'} {\rm Cov}_{K, k'}  - \tanh h_1 {\rm Cov}_{K,1} \Big | \nonumber \\
	& \qquad \qquad  \qquad \qquad \qquad \qquad\leq \int_{\underline \alpha}^{\bar \alpha} d\alpha \,\Big | \sum_{k'=1}^{\infty} \frac{(-1)^{k'+1}}{k'} (\tanh h_1)^{k'} {\rm Cov}_{K, k'} \Big | + \tanh h_1 \int_{\underline \alpha}^{\bar \alpha} d\alpha  \,| {\rm Cov}_{K,1}  |\,. \label{eq:Qp-induction:case-II:4}
\end{align}
Then we proceed as
\begin{align}
	&\frac{(\tanh h_1)^{2k+1}}{2k+1} \int_{\underline h}^{\bar h} dh_0 \int_{\underline \alpha}^{\bar \alpha}d\alpha \, | {\rm Cov}_{K, 2k+1} | \nonumber \\
	& = \int_{\underline h}^{\bar h} dh_0  \int_{\underline \alpha}^{\bar \alpha} d\alpha \, \Big | \Big ( \frac{(\tanh h_1)^{2k}}{2k} {\rm Cov}_{K, 2k} - \frac{(\tanh h_1)^{2k+1}}{2k+1} {\rm Cov}_{K, 2k+1} \Big ) - \frac{(\tanh h_1)^{2k}}{2k} {\rm Cov}_{K, 2k} \Big | \nonumber \\
	& \leq \int_{\underline h}^{\bar h} dh_0 \int_{\underline \alpha}^{\bar \alpha} d\alpha \,\Big | \frac{(\tanh h_1)^{2k}}{2k} {\rm Cov}_{K, 2k} - \frac{(\tanh h_1)^{2k+1}}{2k+1} {\rm Cov}_{K, 2k+1} \Big | + \frac{(\tanh h_1)^{2k}}{2k} \int_{\underline h}^{\bar h} dh_0 \int_{\underline \alpha}^{\bar \alpha}d\alpha \,  | {\rm Cov}_{K, 2k} | \nonumber \\
	& \leq \int_{\underline h}^{\bar h} dh_0 \int_{\underline \alpha}^{\bar \alpha} d\alpha\, \Big | \sum_{k'=2}^{\infty} \frac{(-1)^{k'+1}}{k'} (\tanh h_1)^{k'} {\rm Cov}_{K, k'} \Big | + \frac{(\tanh h_1)^{2k}}{2k} \int_{\underline h}^{\bar h} dh_0 \int_{\underline \alpha}^{\bar \alpha}d\alpha \, | {\rm Cov}_{K, 2k}  | \label{eq:Qp-induction:case-II:5} \\
	& \leq \int_{\underline h}^{\bar h} dh_0 \int_{\underline \alpha}^{\bar \alpha} d\alpha \,\Big | \sum_{k'=1}^{\infty} \frac{(-1)^{k'+1}}{k'} (\tanh h_1)^{k'} {\rm Cov}_{K, k'} \Big | + \tanh h_1 \int_{\underline h}^{\bar h} dh_0 \int_{\underline \alpha}^{\bar \alpha} d\alpha \, | {\rm Cov}_{K,1}  | \nonumber \\
	& \qquad \qquad \qquad\qquad + \frac{(\tanh h_1)^{2k}}{2k} \int_{\underline h}^{\bar h} dh_0 \int_{\underline \alpha}^{\bar \alpha}d\alpha \,| {\rm Cov}_{K, 2k}  | \label{eq:Qp-induction:case-II:6} \\
	& \leq  \frac{\sqrt{5C_P + 40 h_1^2}}{n^{(\theta-1/2)/3}} + \frac{\sqrt{5C_P+40}}{n^{(\theta-1/2)/3}} + M_{2k} \frac{\sqrt{5C_P+40}}{n^{(\theta-1/2)/3}} \label{eq:Qp-induction:case-II:7}
	\\ &
	\leq 
	(M_{2k} + 2) \frac{\sqrt{5C_P+40}}{n^{(\theta-1/2)/3}}\nonumber
\end{align}
where \eqref{eq:Qp-induction:case-II:5} follows from \eqref{eq:Qp-induction:case-II:3}, then \eqref{eq:Qp-induction:case-II:6} follows from \eqref{eq:Qp-induction:case-II:4}, 
and finally \eqref{eq:Qp-induction:case-II:7} follows from Lemma~\ref{thm:concentration-Qp:sum}, \eqref{CovK1} and \eqref{endfirststep}. Summarizing,
\begin{align*}
 \int_{\underline h}^{\bar h} dh_0 \int_{\underline \alpha}^{\bar \alpha}d\alpha \, | {\rm Cov}_{K, 2k+1} | 
 \leq \frac{(2k+1)M_{2k+1}}{(\tanh h_1)^{2k+1}}\frac{\sqrt{5C_P +40}}{n^{(\theta-1/2)/3}}
\end{align*}
with $M_{2k+1}=M_{2k} + 2$, which ends the induction argument.
\end{proof}

\subsection{Proof of Theorem \ref{thm:multiconc}}\label{sec:combiningConc}
We finally show how to combine all the concentration results we obtained in order to prove the 
following theorem. This theorem is a mild variant of Theorem \ref{thm:multiconc}. Inequality \eqref{97} below is exactly Theorem \ref{thm:multiconc}.
%
\begin{theorem}[Overlap concentration]\label{thm:over_conc_long}
Assume the Hamiltonian $\sH$ given by \eqref{Htot} satisfies both GKS inequalities and also that Condition~\ref{conditionC} holds. 
Then for any moment $p\geq 2$, $[\underline h, \bar h]\subset (0,1)$, $[\underline\alpha, \bar\alpha]\subset (0,1)$, $h_1\in (0,1]$, $\theta \in (1/2, 7/8]$,
\begin{align*}
\int_{\underline h}^{\bar h} dh_0 \int_{\underline \alpha}^{\bar \alpha} d\alpha\, \mathbb{E} \big\< \big | Q_k^p - \mathbb{E} [ \< Q_k \> ]^p  \big | \big\>
	\leq
	\frac{2pk}{(\tanh h_1)^{k/2}} \frac{(5C_P + 40)^{1/4}}{n^{(\theta-1/2)/6}}
	\,.
\end{align*}
\label{thm:full-concentration-Qp}
\end{theorem}
\begin{proof}
We integrate both sides of \eqref{eq:full-concentration-Qp:1} over $h_0$ and $\alpha$. As all the square terms are bounded by $1$, by Fubini's theorem we 
are free to exchange the order of the integrals. Theorem \ref{thm:thermalconc}, and Lemmas \ref{thm:concentration-Poisson} and \ref{thm:Qp-induction} are applied accordingly and lead to the estimate
(for any $k\geq 1$)
\begin{align}
\int_{\underline h}^{\bar h} dh_0 \int_{\underline \alpha}^{\bar \alpha} d\alpha\, \mathbb{E}\big\< \big(Q_k - \mathbb{E} \< Q_k \> \big)^2 \big\>
	& \leq \frac{2k}{n} + \frac{(\bar \alpha^2 -\underline \alpha^2) k h_1}{2n^{1-\theta}} + \frac{kM_k}{(\tanh h_1)^k}\frac{\sqrt{5C_P+40}}{n^{(\theta-1/2)/3}} \nonumber \\
	& \leq  \frac{4k^2}{(\tanh h_1)^k}\frac{\sqrt{5C_P+40}}{n^{(\theta-1/2)/3}}
	\label{97}
\end{align}
using $\theta \in (1/2, 7/8]$ (the $7/8$ is enforced by $(\theta-1/2)/3 \le 1-\theta$), $[\underline\alpha, \bar\alpha]\subset (0,1)$, $h_1\in (0,1]$ and $M_k< 3k/2$.
Finally, observe that
\begin{align*}
\mathbb{E} \big\< \big | Q_k^p - \mathbb{E} [ \< Q_k \> ]^p  \big | \big\> 
	& = \mathbb{E}  \Big\< \Big |  ( Q_k - \mathbb{E} \< Q_k \> ) \sum_{l=0}^{p-1}Q_k^{p-1-l} \mathbb{E} [\< Q_k \>]^{l} \Big | \Big\> \leq p \,\mathbb{E}  \big\< \big | Q_k - \mathbb{E} \< Q_k \> \big | \big\> \,.
\end{align*}
By the Cauchy-Schwarz inequality we then have
\begin{align*}
\int_{\underline h}^{\bar h} dh_0 \int_{\underline \alpha}^{\bar \alpha} d\alpha\, \mathbb{E}  \big\< \big | Q_k^p - \mathbb{E} [ \< Q_k \> ]^p  \big | \big\> 
	& \leq p \,\Big\{ \int_{\underline h}^{\bar h} dh_0 \int_{\underline \alpha}^{\bar \alpha} d\alpha \Big\}^{1/2} \Big \{ \int_{\underline h}^{\bar h} dh_0 \int_{\underline \alpha}^{\bar \alpha} d\alpha \,\mathbb{E}  \big\< \big( Q_k - \mathbb{E} \< Q_k \> \big)^2 \big\>  \Big \}^{1/2}
\end{align*}
which ends the proof once combined with \eqref{97}.
\end{proof}

%

%
%
%
%
\appendix
\section{Some technicalities}
\subsection{Proof of the approximation inequality \eqref{approxinequ}}\label{proof3.3}

Note that
\begin{align*}
\big | p_{n}(h_0, h_1, \alpha) - p_{n}(0,0,0) \big |
	& = \big | p_{n}(h_0, h_1, \alpha) - p_{n}(0,h_1,0) \big | \\
	& \leq \big | p_{n}(h_0, h_1, \alpha) - p_{n}(0, h_1, \alpha) \big | + \big | p_{n}(0, h_1, \alpha) - p_{n}(0, h_1, 0) \big |\,.
\end{align*}
We have $| \frac{d p_{n}(h_0, h_1, \alpha)}{d{h_0}}| = |\E\<Q_1\>| \leq 1$ and from \eqref{68} 
we also have $| \frac{d p_n(0, h_1, \alpha)}{d\alpha}| \leq h_1 n^{-(1-\theta)}$. Thus by the mean value theorem we obtain \eqref{approxinequ}, i.e. $|p_n(h_0,h_1,\alpha) - p_n(0,0,0)| \leq h_0 + \alpha h_1n^{-(1-\theta)}$.
\subsection{A property of the Poisson distribution}\label{poisson-appendix}


Any function $g: \mathbb{N} \to \mathbb{R}$ of a random variable $X\sim{\rm Poi}(\nu)$ with 
Poisson distribution and mean $\nu$, and such that $\mathbb{E}\,g(X)$ exists, satisfies
\begin{align}
\frac{d \, \mathbb{E}\,g(X)}{d \nu} & = \sum_{k=0}^{\infty} \frac{d}{d\nu}\bigg\{\frac{\nu^{k} e^{-\nu}}{k!}\bigg\} g(k) 
= \sum_{k=1}^{\infty} \frac{\nu^{k-1} e^{-\nu}}{(k-1)!} g(k) - \sum_{k=0}^{\infty} \frac{\nu^k e^{-\nu}}{k!} g(k)
\nonumber \\ &
=
\sum_{k=0}^{\infty} \frac{\nu^{k} e^{-\nu}}{k!} g(k + 1) - \sum_{k=0}^{\infty} \frac{\nu^k e^{-\nu}}{k!} g(k)
=
\mathbb{E}\,g(X+1) - \mathbb{E}\, g(X)\,.
\label{eq:dPoissonMean}
\end{align}


%
%

%
%
\subsection{Multivariate Harris inequality}\label{appendix-harris}
For completeness we provide here a simple proof of the multivariate version of the Harris inequality. We refer to \cite{Liggett:2011} for more information. 

\begin{lemma}[Multivariate version of the Harris inequality]
Let $g,\tilde{g}:\mathbb{R}^n\mapsto \mathbb{R}$ be two functions of the random vector $\bx= (x_1, \ldots, x_n)$ where all components are independent random variables. If for all $i \in \{1, \dots, n\}$ $g$ and $\tilde{g}$ are both monotone w.r.t. $x_i$ with same monotonicity, i.e. $\partial_{x_i}g(\bx)\,\partial_{x_i}\tilde{g}(\bx)\ge 0 \ \forall \ i$, then $\mathbb{E}[g(\bx)\, \tilde{g}(\bx)] - \mathbb{E}\,g(\bx) \, \mathbb{E}\,\tilde{g}(\bx) \geq 0$.
\label{thm:Harris}
\end{lemma}
\begin{proof}
Let $\bx_i^j \equiv (x_i, x_{i+1}, \dots, x_{j})$.
The monotonicity w.r.t. $x_1$ implies
\begin{align*}
\mathbb{E}_{x_1} \mathbb{E}_{x'_1} \big [ \big ( g(x_1, \bx_2^n) - g(x'_1, \bx_2^n) \big ) \big ( \tilde{g}(x_1, \bx_2^n) - \tilde{g}(x'_1, \bx_2^n) \big ) \big ] \geq 0
\end{align*}
which by expanding the product can be simplified to
\begin{align*}
\mathbb{E}_{x_1}[g(\bx) \,\tilde{g}(\bx)]-\mathbb{E}_{x_1}g(\bx)\, \mathbb{E}_{x_1}\tilde{g}(\bx) \geq 0\,.
\end{align*}
The proof then proceeds by induction. Suppose 
\begin{align}
\mathbb{E}_{\bx_1^{i-1}}[g(\bx)\, \tilde{g}(\bx)]-\mathbb{E}_{\bx_1^{i-1}}g(\bx)\, \mathbb{E}_{\bx_1^{i-1}}\tilde{g}(\bx)\label{hypo}
	\geq 0\,.	
\end{align}
Again, the monotonicity w.r.t. $x_i$ implies 
\begin{align*}
&\mathbb{E}_{x_i} \mathbb{E}_{x'_i} \big [ \big ( \mathbb{E}_{\bx_1^{i-1}}  g(\bx_1^{i-1}, x_i, \bx_{i+1}^n)  - \mathbb{E}_{\bx_1^{i-1}}  g(\bx_1^{i-1}, x'_i, \bx_{i+1}^n)  \big ) \nonumber\\
&\qquad \qquad \qquad\cdot\big ( \mathbb{E}_{\bx_1^{i-1}}  \tilde{g}(\bx_1^{i-1}, x_i, \bx_{i+1}^n)  - \mathbb{E}_{\bx_1^{i-1}}  \tilde{g}(\bx_1^{i-1}, x'_i, \bx_{i+1}^n)  \big ) \big ] \geq 0
\end{align*}
which can be simplified to
\begin{align}
\mathbb{E}_{x_i} \big [ \mathbb{E}_{\bx_1^{i-1}} g(\bx) \, \mathbb{E}_{\bx_1^{i-1}}  \tilde{g}(\bx) \big ]-\mathbb{E}_{\bx_1^{i}}g(\bx)\, \mathbb{E}_{\bx_1^{i}}\tilde{g}(\bx) \geq 0\,.
\label{eq:harris:3}
\end{align}
The induction is ended by noting that with the hypothesis \eqref{hypo} the identity \eqref{eq:harris:3} can further be relaxed to
\begin{align*}
\mathbb{E}_{\bx_1^{i}}[g(\bx)\, \tilde{g}(\bx)]-\mathbb{E}_{\bx_1^{i}}g(\bx)\, \mathbb{E}_{\bx_1^{i}}\tilde{g}(\bx) \geq 0\,.
\end{align*}
This ends the induction argument and the proof.
\end{proof}
\section{On the concentration and existence of the pressure} \label{sec:free_ent_conc}
We consider Hamiltonian \eqref{mostgeneralhamiltonian} with independent random couplings $J_X$, $X\subset \{1,\dots, n\}$ and prove the following generic result used in \eqref{concpressuregeneral}. We then discuss a simple argument and condition that guarantees the existence of the thermodynamic limit using the first GKS inequality. 
We verify that these results apply to Example \ref{ex1}.

\begin{proposition}[Concentration of the pressure]
\label{thm:free_ent_conc}
Let $J_X$, $X\subset \{1,\dots, n\}$ be independent random variables such that 
$\sum_{X\subset \{1,\dots, n\}} \mathrm{Var}(J_X) \leq C_P n$
for some numerical constant $C_P>0$.
Then we have
$\mathbb{E}[(P_n - p_n)^2] \leq C_P/n$.
\end{proposition}

\begin{proof}
The proof is a simple application of the Efron-Stein inequality. Set $\bm{J}\equiv (J_X, X\subset\{1,\dots,n\})$.
Let $\bm{J}^{(X)}$ be a vector such that $\bm{J}^{(X)}$ differs from $\bm{J}$ only at the $X$--th component which becomes $J_X^\prime$ drawn 
independently from the same distribution as the one of $J_X$ (note that the random variables $J_X$ for different $X$ do not necessarily have the same distribution). Efron Stein's inequality tells us that
\begin{align}
\mathbb{E} \big [ ( P_n - \mathbb{E}\,P_n )^2 \big ]
	& \leq \frac{1}{2} \sum_{X\subset \{1,\dots,n\}} \mathbb{E}_{\bm{J}\setminus J_X} \mathbb{E}_{J_X}\mathbb{E}_{J'_X} \big [ ( P_n(\bm{J}) - P_n(\bm{J}^{(X)}) )^2 \big ] \,.
	\label{eq:free_ent_conc:ef1}
\end{align}
An elementary interpolation gives
\begin{align*}
\big | P_n(\bm{J}) - P_n(\bm{J}^{(X)})\big |
	& = \frac{1}{n} \Big | \int_0^1 ds \frac{d}{ds} \ln \sum_{\bm{\sigma} \in \{\pm 1 \}^n} \exp \big \{- \sH_0(\bm{\sigma},\bm{J}^{(X)}) + s (J_X -J'_X) \sigma_X \big \} \Big | \nonumber \\
	& = \frac{1}{n} \Big | \int_0^1 ds\, (J_X - J'_X)  \< \sigma_X \>_s  \Big | \leq \frac{1}{n} | J_X - J'_X | \,.
\end{align*}
Replacing in \eqref{eq:free_ent_conc:ef1} (and recalling $p_n\equiv \mathbb{E}\, P_n$) gives
\begin{align*}
\mathbb{E}\big[(P_n - p_n)^2\big]
	 \leq \frac{1}{2n^2} \sum_{X\subset\{1,\dots,n\}} \mathbb{E}_{J_X} \mathbb{E}_{J'_X} \big [ ( J_X - J'_X )^2 \big ] 
	= \frac{1}{n^2}\sum_{X\subset\{1,\dots,n\}}\mathrm{Var}(J_X)\,.
\end{align*}
With the hypothesis on $\mathrm{Var}(J_X)$ the proof is complete.
\end{proof}

An easy and more or less standard superadditivity argument proves that the thermodynamic limit exists for the ferromagnetic model \eqref{mostgeneralhamiltonian}. 
We give the argument for completeness. 
For simplicity we consider that there exists
a maximal size $x_{\rm max}$ independent of $n$ such that $\vert X\vert \leq x_{\rm max}$. We suppose furthermore that all $J_X$ are independent with a distribution that depends only on the cardinalities 
$\vert X\vert$ (in other words given a cardinality they are i.i.d.) and also 
\begin{align}\label{conditionthermo}
 \frac1n\sum_{X\subset \{1, \dots, n\}}\mathbb{E}\,J_X  = \frac1n\sum_{\vert X\vert =1}^{x_{\rm max}}\binom{n}{\vert X\vert} m(\vert X\vert) \leq C
\end{align}
 where $m({\vert X\vert}) \equiv \mathbb{E}\,J_X$ and $C$ a positive constant independent of $n$. 
 
\begin{proposition}[Existence of the thermodynamic limit of the pressure]\label{propD2}
Let $J_X, X \subset \{1,\dots,n\}$ be independent random variables with a probability distribution supported on $\mathbb{R}_{\ge 0}$ depending only on 
$\vert X\vert$. Moreover assume  
$J_X=0$ for $\vert X\vert > x_{\rm max}$ independent of $n$. Let \eqref{conditionthermo} be satisfied. Then $\lim_{n \rightarrow +\infty} p_n$ exists and is finite.
\end{proposition}
\begin{proof}
Fix non-zero integers $n_1$, $n_2$ both greater than $x_{\rm max}$ and $n \equiv n_1 + n_2$. Consider a set of realizations $S\equiv \{J_X, X\subset\{1,\dots,n\}\}$. This set can be split in three 
disjoint sets $S=S_1\cup S_2\cup S_{12}$ with
$S_1 \equiv \{J_X, X\subset\{1,\dots, n_1\}\}$, $S_2 \equiv \{J_X, X\subset\{n_1+1,\dots, n\}\}$ and $S_{12} \equiv \{J_X, X\cap\{1,\dots, n_1\} \neq \emptyset, X\cap\{n_1+1,\dots, n\} \neq \emptyset\}$.
Let $\ln {\cal Z}(S)/n$ the pressure corresponding to the Hamiltonian with couplings in $S$ and 
$\ln {\cal Z}(S_1)/n_1$ and $\ln {\cal Z}(S_2)/n_2$ the pressures corresponding to the Hamiltonians with couplings from $S_1$ and $S_2$ only. One can show, using the first GKS inequality, that\footnote{Interpolating from $t=0$ to $1$ over $\ln \sum_{\bsig\in\{\pm1\}^n}\exp(-{\cal H}_t(\bsig))$ with ${\cal H}_t(\bsig)\equiv-\sum_{X:J_X\in S_1\cup S_2} J_X \sigma_X-t\sum_{X:J_X\in S_{12}} J_X \sigma_X$ shows that $\ln {\cal Z}(S) = \ln {\cal Z}(S_1) + \ln {\cal Z}(S_2)+\sum_{X:J_X\in S_{12}}J_X\int_0^1dt\langle \sigma_X\rangle_t$ (using that $S_1$ and $S_2$ are disjoint). Then the first GKS inequality gives $\langle \sigma_X\rangle_t\ge 0$. As $J_X\ge 0$ too, we obtain the result.}
\begin{align*}
 \ln {\cal Z}(S) \geq \ln {\cal Z}(S_1) + \ln {\cal Z}(S_2)\,.
\end{align*}
Then averaging over all coupling constants in $S$, using that they are independent with distributions depending only on the cardinality $\vert X\vert$ and 
that all cardinalities are contained in $S$, $S_1$ and $S_2$, we obtain 
\begin{align*}
 \mathbb{E}_S\ln {\cal Z}(S) \geq \mathbb{E}_{S_1}\ln {\cal Z}(S_1) + \mathbb{E}_{S_2}\ln {\cal Z}(S_2)
\end{align*}
which is equivalent to $np_n \geq  n_1 p_{n_1} + n_2 p_{n_2}$ (for $n_1$, $n_2$ greater than $x_{\rm max}$).
This means that the function $n\mapsto np_n$ is a superadditive sequence and therefore by Fekete's lemma the limit $\lim_{n\to +\infty} p_n$ equals $\sup_n p_n$. To show that 
$\sup_n p_n$ is finite note that 
\begin{align*}
 p_n \leq - \frac1n \,\mathbb{E}\,\min_{\bm{\sigma}}\mathcal{H}(\bm{\sigma}) + \ln 2 = \frac1n \sum_{X\subset\{1,\dots,n\}} \mathbb{E}\,J_X + \ln 2 \leq C+ \ln2
\end{align*}
using $J_X\ge 0$ and condition \eqref{conditionthermo}. This ends the proof.
\end{proof}

Consider now Example~\ref{ex1} for $n$ large and $p$ fixed. We have $J_X = 0$ for all subsets with cardinalities $\vert X\vert$ different from $1$ and $p$. For $\vert X\vert =1$ the coupling constants $J_X=H$ are deterministic so obviously $\mathrm{Var}(J_X) =0$. For $\vert X\vert =p$ the couplings $J_X$ are independent Bernoulli variables taking 
value $J$ with probability $\gamma n \binom{n}{p}^{-1}$ and $0$ with complementary probability, so $\mathrm{Var}(J_X) = J^2 \gamma n \binom{n}{p}^{-1} \big ( 1 - \gamma n \binom{n}{p}^{-1} \big )$. Thus
\begin{align*}
 \sum_{X\subset\{1,\dots,n\}}\mathrm{Var}(J_X) = \binom{n}{p} J^2 \gamma n \binom{n}{p}^{-1} \Big ( 1 - \gamma n \binom{n}{p}^{-1} \Big ) < J^2\gamma n\,.
\end{align*}
Therefore Proposition \ref{thm:free_ent_conc} applies. Similarly the condition for the existence of the thermodynamic limit of the pressure is also met because the left hand side of \eqref{conditionthermo} equals
\begin{align*}
 \frac{1}{n}\sum_{X\subset\{1,\dots,n\}}\mathbb{E}\,J_X = H + J\gamma\,.
\end{align*}
The mixed $p$-spin models can be treated similarly.

\section*{Acknowledgments}
This work was done while J.B. was affiliated with EPFL. J.B. and C.L.C. acknowledge funding from the Swiss National Science
Foundation grant 200021-156672. J.B.  would like to particularly 
thank Amin Coja-Oghlan for suggesting him to work on this problem and for the very many interesting discussions 
they had when he visited Dr. Coja-Oghlan's group in Francfurt, as well as Florent Krzakala for insightful comments, Adriano Barra for pointing relevant references and Silvio Franz for clarifications on some of his works. J.B. also deeply thanks Nadia Bersier for her support.
%
{
	\singlespacing
	\bibliographystyle{unsrt_abbvr}
	\bibliography{refs}
}
\end{document}